\documentclass{llncs}
\usepackage{llncsdoc}
\usepackage{amsfonts,amsmath,setspace}
\usepackage[usenames,dvipsnames]{xcolor}
\usepackage{afterpage}
\usepackage{semantic}
\usepackage{stmaryrd}
\usepackage{paralist}
\usepackage{float}
\usepackage{verbatim}

\usepackage{graphicx}
\usepackage{amssymb}

\usepackage{pgffor}

\usepackage{times}
 
\newcommand{\suchthat}{\mid}

\newcommand{\NNF}{{\sf NNF}}

\newcommand{\SC}{{\sf sc}}

\newcommand{\restfun}[2]{ #1{\restriction}_{#2}}

\newcommand{\mfun}{{\sf m}}
\newcommand{\eqm}{=_{\mfun}}

\newcommand{\eql}{\approx}
\newcommand{\neql}{\not \approx}
\newcommand{\dom}{{\sf dom}}
\newcommand{\range}{{\sf range}}

\newcommand{\statement}{S}

\newcommand{\ALC}{\mathcal{ALC}}
\newcommand{\ALCM}{\mathcal{ALCM}}

\newcommand{\KbALC}[2]{ (#1,#2)}
\newcommand{\KB}[3]{( #1,#2,#3 )}

\newcommand{\Kb}{\mathcal{K}}

\newcommand{\Tb}{\mathcal{T}}
\newcommand{\Ab}{\mathcal{A}}
\newcommand{\Mb}{\mathcal{M}}

\newcommand{\interp}{\mathcal{I}}

\newcommand{\forest}{{\cal L}}

\newcommand{\Ts}{\Tb}
\newcommand{\As}{\Ab}
\newcommand{\Xs}{\mathcal{X}}
\newcommand{\Zs}{\mathcal{Z}}
\newcommand{\Ys}{\mathcal{Y}}

\newcommand{\judgement}{J}

\newcommand{\basejuicio}{base judgement}
\newcommand{\varjuicio}{variable judgement}
\newcommand{\basenode}{base  node}
\newcommand{\varnode}{variable node}

\newcommand{\andnode}{\boldsymbol{\wedge}  }
\newcommand{\bigandnode}{\boldsymbol{\bigwedge}  }
\newcommand{\ornode}{\boldsymbol{\vee}  }

\newcommand{\andorgraph}{ \mathbb{G} }
\newcommand{\andorgraphP}{ \mathbb{G}' }

\newcommand{\derconcepts}[2]{\KbALC{#1}{#2}}
\newcommand{\derboxes}[3]{ \KB{#1}{#2}{#3}   }
\newcommand{\derbot}{\boldsymbol{\bot}}

\newcommand{\circular}{{\sf circular}}

\newcommand*{\Scale}[2][4]{\scalebox{#1}{$#2$}}

\newcommand{\rel}{\prec}

\newcommand{\botR}{(\bot)}  
\newcommand{\andR}{(\sqcap)}  
\newcommand{\orR}{(\sqcup)}  
\newcommand{\transR}{({\sf trans})}  
\newcommand{\botOnePR}{(\bot_1)}  
\newcommand{\botTwoPR}{(\bot_2)}  
\newcommand{\botThreePR}{(\bot_3)}  
\newcommand{\andPR}{(\sqcap^{\prime})}  
\newcommand{\orPR}{(\sqcup^{\prime})}  
\newcommand{\forallPR}{(\forall)}
\newcommand{\transPR}{({\sf trans}^{\prime})}
\newcommand{\closePR}{({\sf close})}
\newcommand{\equal}{(=)}
\newcommand{\difference}{(\not=)}

\newcommand{\setRoles}{{\bf{ R}}}

\newcommand{\T}{ \mathbb{S} }
\newcommand{\setS}{\Delta}
\newcommand{\lab}{{\cal L}}
\newcommand{\Eps}{ \mathcal{E} }

\newcommand{\graph}[1]{#1-structure}
\newcommand{\unfold}{{\sf set}}

\newcommand{\maxlength}[1]{{\sf maxl}^{#1}}
\newcommand{\maxprec}{\maxlength{\prec}}

\newcommand{\verteq}{\rotatebox{90}{$\,=$}}
\newcommand{\equalto}[2]{\underset{\scriptstyle\overset{\mkern4mu\verteq}{#2}}{#1}}

\title{Complexity of  the Description Logic $\ALCM$}

\author{ M\'onica Mart\'inez \inst{2} \and Edelweis Rohrer\inst{2} \and Paula Severi\inst{1} }

\authorrunning{Mart\'inez, Rohrer, Severi}

\institute{Department of Computer Science, University of Leicester, England \\
\and
Instituto de Computaci\'on, Facultad de Ingenier\'ia, \\  Universidad de la Rep\'ublica, Uruguay\\
}

\begin{document}

\maketitle

\parindent=0cm

\begin{abstract}
In this paper we show that the problem of checking consistency of a knowledge base in 
the Description Logic $\ALCM$ is ExpTime-complete.
The $\mathcal{M}$ stands for  meta-modelling as defined by Motz, Rohrer and Severi.
To show our main result, we define an ExpTime Tableau algorithm   
as an extension of an algorithm for checking consistency of a knowledge base in 
$\ALC$ by Nguyen and Szalas.
\end{abstract}

\section{Introduction}
The main motivation of the present work is to
 study the complexity of 
  meta-modelling as defined in
   \cite{Motz2015,DBLP:conf/jist/MotzRS14}.	
  No study of complexity has been done so far
  for this approach and we would like to analyse if 
  it  increases the complexity of a given description logic.
  \\
  It is well-known that 
  consistency of a (general) knowledge base in  $\ALC$ is 
  ExpTime-complete. 
The hardness result was proved in~\cite{DBLP:conf/ijcai/Schild91}.
A matching upper bound for $\ALC$ was given by De Giacomo and Lenzerini by a
reduction to PDL \cite{DBLP:conf/dlog/GiacomoL96}.
\\
In this paper, we show  that the consistency  of a knowledge base in
$\ALCM$ is ExpTime-complete where the $\mathcal{M}$
  stands for the 
  meta-modelling approach mentioned above.
Hardness follows trivially 
from the fact that $\ALCM$ is an extension of $\ALC$ since
any algorithm that decides consistency of a knowledge base 
in $\ALCM$ can be used
to decide consistency of a knowledge base in $\ALC$.
In order to give a matching upper bound on the complexity of this 
problem, it is enough to show that there is a particular algorithm
with running time at most $O(2^n)$ where $n$ is the size of
the knowledge base.
The standard  tableau algorithm for $\ALC$ 
which builds completion trees, e.g. see 
 \cite{DBLP:conf/dlog/2003handbook}, 
can be extended with 
the expansion rules for meta-modelling 
of \cite{Motz2015,DBLP:conf/jist/MotzRS14}.  
This algorithm is the bases for Semantic Web reasoners 
such as Pellet \cite{DBLP:journals/ws/SirinPGKK07}.  
 However, it 
has a high (worse case) complexity, namely NExpTime,
 and cannot be used to prove
that the consistency problem for $\ALCM$ is ExpTime-complete.
\\
Other approaches to meta-modelling use
 translations to prove decidability and/or
complexity \cite{Motik07,DBLP:conf/owled/PanHS05,DBLP:conf/semweb/GlimmRV10,DBLP:journals/ijsi/JekjantukGP10,DBLP:conf/aaai/GiacomoLR11,DBLP:conf/dlog/HomolaKSV13,DBLP:conf/dlog/HomolaKSV14,Lenzerinietal2014}.
However,  translations do not seem to work for $\ALCM$ 
due to the combination of 
 a flexible syntax 
 with a strong semantics of well-founded sets.
 A consistency  algorithm for $\ALCM$  
 has to check if
   the domain of the canonical model under construction is a
     well-founded set.
     This is an unusual and  interesting aspect of our approach but at
     the same time what makes it more difficult to solve.
\\
The contributions of this paper are the following: 
\begin{enumerate}

\item  We define a  tableau algorithm for checking
consistency of a knowledge base in $\ALCM$
as an extension of an algorithm   
for $\ALC$ by Nguyen and Szalas~\cite{Nguyen2009}.

\item  We prove correctness and show that the complexity of 
our algorithm for $\mathcal{ALCM}$ is  
 ExpTime. 
 
 \item
 From the above two items, 
we obtain the main result of our paper which is the fact that
the problem of checking consistency of a knowledge base in $\ALCM$ is ExpTime-complete. 
 \end{enumerate}
\noindent
Hence, in spite of the fact that our algorithm has 
the burden of having to check for well-founded sets,
complexity does not change when moving from $\ALC$ to $\ALCM$.

\section{A Flexible Meta-modelling Approach for Re-using Ontologies}

A knowledge base in $\ALCM$ contains an Mbox besides of a Tbox and an Abox.
 An  Mbox is a set of  equalities of the form
$a \eqm A$ where $a$ is an individual and $A$ is a concept~\cite{Motz2015,DBLP:conf/jist/MotzRS14}.  
Figure \ref{fig:firstView} shows an example of two ontologies  separated by a horizontal line. 
 The two ontologies conceptualize the same entities at different levels of granularity. 
In the  ontology above the horizontal line, rivers and lakes are formalized as  individuals while in the one below the line they are concepts. If we want to integrate these ontologies into a single ontology
(or into an ontology network) it is necessary to interpret the individual $river$ and the concept $River$ as the same real object. Similarly for $lake$ and $Lake$.
The Mbox for this example contains two equations:
\[
\begin{array}{ll}
river \eqm River & \ \ \ 
lake \eqm Lake
\end{array}
\]
 These equalities are called  \emph{meta-modelling axioms} and in this case, 
 we say that the ontologies are related through {\em meta-modelling}.
In Figure \ref{fig:firstView}, meta-modelling axioms are represented by  dashed edges.
After adding the meta-modelling axioms for rivers and lakes, the concept  $HydrographicObject$  is now also a {\em meta-concept} because it is a concept that contains an individual which is also a concept. \\
\begin{figure}
\centering
\includegraphics[width=0.5\linewidth]{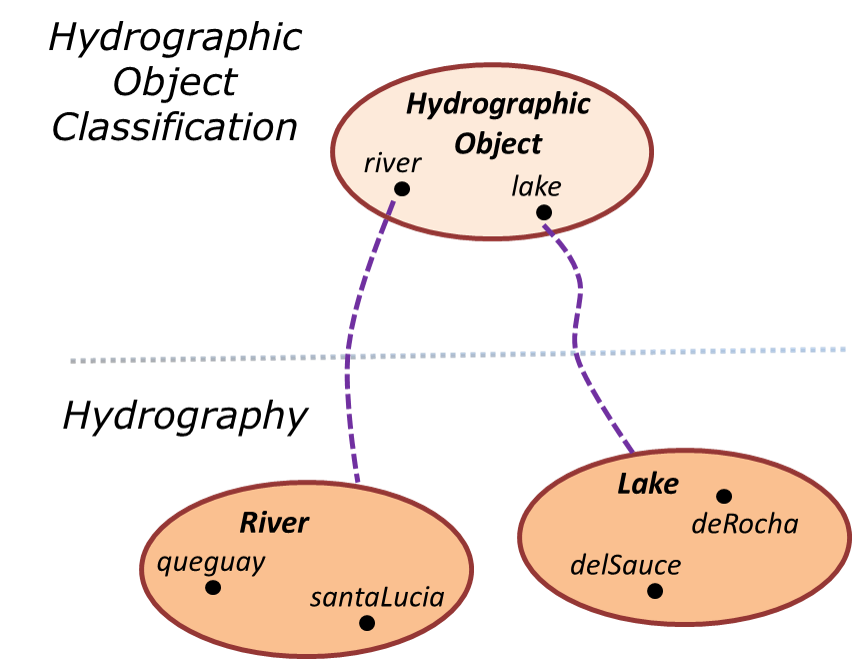}
\caption{Two ontologies  on Hydrography  }
\label{fig:firstView}
\end{figure}
This kind of meta-modelling  can be expressed in OWL Full but it cannot be expressed in OWL DL.
The fact that it is expressed in OWL Full is not very useful since the meta-modelling provided by OWL Full is so expressive that leads to undecidability \cite{Motik07}. 
OWL 2 DL has a very restricted form of meta-modelling called {\em punning} where the same identifier can be used as an individual and as a concept \cite{FOST}. These identifiers are treated
as different objects by the reasoner and it is not possible  to detect certain inconsistencies.
 We next illustrate two examples where OWL would not detect inconsistencies because the identifiers, though they look syntactically equal, are actually  different. 
\begin{example}
 \label{example:introductionRiver}
If we introduce 
 an axiom expressing that  \emph{HydrographicObject} is a subclass of \emph{River}, 
then  OWL's reasoner will not detect that 
the interpretation of $River$ is not a well founded set
(it is a set that  belongs to itself). 
\end{example}
\begin{example}
\label{example:introductionSimpletransference}
We add  two axioms, the first one says that
 $river$ and $lake $ as individuals are equal
and the second one says that the classes $River$ and $Lake$ are
disjoint.
 Then 
OWL's reasoner  does not detect that there is a contradiction.
\end{example}
In order to detect these inconsistencies, $river$ and $ River$ should be made
 semantically equal, i.e. the interpretations of the individual $river$ and 
 the concept $River$  should be the same. 
 The domain of an interpretation cannot longer consists of  only  basic objects 
 but it must be any well-founded set. 
The well-foundness of our model is not ensured by means of fixing layers beforehand as in \cite{DBLP:conf/owled/PanHS05,DBLP:journals/ijsi/JekjantukGP10}
 but it is the reasoner which  checks for circularities. 
 This approach allows the user to have any number of levels or layers 
 (meta-concepts, meta meta-concepts and so on). 
The user does not have to write or know the layer of the concept because the reasoner will infer it for him. In this way, axioms can also naturally mix elements of different layers and the user has the flexibility of changing the status of an individual at any point without having to make any substantial change to the ontology. 

\section{The Description Logic $\ALC$}
\label{Section:DLALC}
In this section we recall the Description Logic
 $\ALC$~\cite{DBLP:journals/ai/Schmidt-SchaussS91,DBLP:conf/dlog/2003handbook}.
We assume we have three pairwise disjoint sets: a set of  individuals, a set of atomic concepts and a set of atomic roles. Individuals are denoted by $a, b, \ldots$, atomic concepts by $A, B, \ldots$  and atomic roles by $R, S, \ldots$. We use $C,D$ to denote arbitrary concepts.
{\em Concepts} are defined by the following grammar:
\[ C, D ::=  A \mid \top  \mid   (\neg C)   \mid  (C \sqcap D) \mid (C \sqcup D)  \mid  (\forall R.C) \mid (\exists R.C) \] 
We omit parenthesis according to the
 following  precedence order of the description logics operators:
 (i) $\neg$,  $\forall$,  $\exists$
   (ii) $\sqcap$, (iii) $\sqcup$. 
    Outermost parenthesis  can sometimes be omitted.
 \\
We use $\bigsqcap \{C_1, \ldots, C_n \}$ to denote $C_1 \sqcap \ldots \sqcap C_n$.
Syntactic equality between concepts or individuals is denoted by $=$.  \\
%
We say that $C$ is a {\em (syntactic) subconcept } of a concept $D$ if
 $C \in \SC(D)$ where $\SC$ is defined as follows.
\[
\begin{array}{lll}
\SC(C) &  =   \{ C \} \mbox{ if  } C \in \{ A, \top,\bot \} \\ 
\SC(\neg C) & =  \SC(C) \cup \{ \neg C \} \\
\SC(C \sqcap D)  & =  \SC(C) \cup \SC(D)  \cup \{C \sqcap D \} \\ 
\SC(C \sqcup D)  & =  \SC(C) \cup \SC(D)  \cup \{C \sqcup D \} \\
\SC(\forall R.C )  & =  \SC(C) \cup \{\forall R.C\} \\ 
\SC(\exists R.C) & =  \SC(C) \cup \{\exists R.C\} \\
\end{array}
\]
%
%

 A  {\em knowledge base} $\Kb$ in  $\ALC$
is a pair $\KbALC{\Tb}{\Ab}$ 
where 
\begin{enumerate}
\item $\Tb$, called a {\em Tbox}, 
is a finite set of axioms of the form $C \sqsubseteq D$, with $C$, $D$ any two concepts.
 Statements of the form 
$C \equiv D$ are abbreviations for  $C \sqsubseteq D$ and  $D\sqsubseteq C$. 
 
\item $\Ab$, called an {\em Abox}, 
is a finite set of statements of the form $C(a)$, $R(a, b)$,  $a=b$ 
 or $a \not = b$. 

\end{enumerate}
The set of all individuals occurring in $\As$ is denoted by $\dom (\As)$. 

\noindent To avoid confusion with the syntactic equality, for the statements of the Abox we always write 
the information of it, i.e.  $a = b \in \Ab$.

\noindent Note that Aboxes contain equalities and inequalities between individuals
in spite of the fact that they are not part of the standard definition of 
$\ALC$.
There are two reasons for adding them. First of all,
 this is a very useful OWL feature. Second and most important,
 it makes it evident that equality and difference
 between individuals play an important role 
 in the presence of meta-modelling since an  equality between individuals is transferred
into an equality  between the corresponding concepts and conversely.
\\
\\
An {\em interpretation} 
$\interp = (\Delta^{\interp}, {\cdot}^{\interp})$
consists of a 
non-empty set $\Delta^{\interp}$ (sometimes we drop the super-index when the name of the interpretation is clear from
the context and write just $\Delta$), called the {\em domain} of $\interp$, and a function 
$\cdot^{\interp}$ which maps every concept to a subset of $\Delta$ and every role to a subset of $\Delta \times \Delta$ such that, for all concepts $C$, $D$ and role $R$
the following equations are satisfied:

\[\begin{array}{lcl}
 \top^{\interp} & = & \Delta^{\interp} \\
 (C \sqcap D)^{\interp} & = & C^{\interp} \cap D^{\interp} \\
 (C \sqcup D)^{\interp} & = & C^{\interp} \cup D^{\interp} \\
 (\neg C) ^{\interp} & = & \Delta \backslash C^{\interp} \\
 (\exists R.C) ^{\interp} & = & \{x \mid \exists y. (x, y) \in R^{\interp}$ and $y \in C^{\interp} \}\\
 (\forall R.C) ^{\interp} & = &\{x \mid \forall y. (x, y) \in R^{\interp}$ implies $y \in C^{\interp} \}\\
\end{array}\]
An interpretation $\interp$ {\em satisfies a concept $C$}, denoted 
 by $\interp \models C$,  if $C^{\interp} \neq \emptyset$
  and it satisfies a set $\Xs$ of concepts, denoted by $\interp \models \Xs$, 
 if 
$(\bigsqcap \Xs)^\interp \not = \emptyset$.
Note that $(\bigsqcap \Xs)^{\interp} = \bigcap_{C \in \Xs} C^{\interp}$.\\  
An  interpretation $\interp$ {\em satisfies a TBox}
 $\Tb$, denoted by $\interp \models \Tb$,
  if $C^{\interp} \subseteq D^{\interp}$ for each $C \sqsubseteq D$ in $\Tb$. 
\\
An  interpretation $\interp$ {\em satisfies a set $\Xs$ of concepts w.r.t. Tbox $\Tb$}
 or $\interp$ {\em satisfies $\KbALC{\Tb}{\Xs}$}, denoted by
 $\interp \models \KbALC{\Tb}{\Xs}$,  if $\interp$ satisfies $\Tb$ and $\Xs$. 
\\
An interpretation $\interp$
  {\em validates a concept C}, denoted as $\interp \models C \equiv \top$, 
   if $C^{\interp} = \Delta^{\interp}$.    
   \\
 An interpretation $\interp$
  {\em validates a set $\Xs$ of concepts} if 
  $\interp$ validates every concept in $\Xs$, or equivalently 
  $\interp \models \bigsqcap \Xs \equiv \top$.   
\\
An  interpretation $\interp$ {\em satisfies an ABox}
 $\Ab$, denoted by $\interp \models \Ab$,
  if $a^{\interp} \in C^{\interp}$ for each $C(a)$ in $\Ab$, 
$( a^{\interp}, b^{\interp} ) \in R^{\interp}$ for each $R(a, b)$ in $\Ab$,
 $a^{\interp} = b^{\interp}$ for each $a = b$ in $\Ab$ 
 and $a^{\interp} \not= b^{\interp}$ for each $a \not= b$ in $\Ab$. 
\\
An  interpretation $\interp$ is a {\em model} of $\KbALC{\Tb}{\Ab}$, denoted
by $\interp \models \KbALC{\Tb}{\Ab}$ if
 it satisfies the Tbox $\Tb$ and the Abox $\Ab$.\\
We say that a knowledge base $\Kb=\KbALC{\Tb}{\Ab}$ is \emph{consistent} 
(or \emph{satisfiable}) if
 there exists a model of $\Kb$.\\
 We say that a concept is in {\em negation normal form } if negation occurs
in front of atomic concepts only.  
The negation normal form of a concept
is denoted by $\NNF(C)$ and defined in 
Figure~\ref{figure:fnn}.  An Abox and a Tbox are 
 also converted into negation normal form.
 An  axiom $C \equiv D$ in $\Tb$ is converted into
   $\NNF(\neg C \sqcup  D) \sqcap \NNF(\neg D \sqcup  C)$.  
   Actually, a concept $C$ in the set $\NNF(\Tb)$ which is in negation normal form
   represents the axiom $C \equiv \top$.
   This means that $x \in C^\interp$ for all $x$ in the domain $\Delta^\interp$ 
   and all $C$ in $\NNF(\Tb)$. 
   Thus,
     an interpretation $\interp$ is a model of $\Tb$ if and only if
    $\interp$ validates every concept $C \in \NNF(\Tb)$. 
When $\Tb$ is in negation normal form, then $\Tb$ is not a set of inclusions
but just a set of concepts $C$ such that $C \equiv \top$ should hold. 
In that case, 
we say that $\interp$ is a model of $\Tb$, 
 denoted by $\interp \models  \bigsqcap  
 \Tb \equiv \top$, 
if  $\interp$ validates every concept $C \in \Tb$. 

\begin{figure}
\[
\begin{array}{lll}
 \NNF(A) = A \mbox{   if $A$ is an atomic concept} \\ 
 \NNF(\neg A) = \neg A \mbox{   if $A$ is an atomic concept} \\ 
 \NNF(\neg \neg C) = \NNF(C) \\ 
 \NNF(C \sqcup D) = \NNF(C) \sqcup  \NNF(D) \\
 \NNF(C \sqcap D) = \NNF(C) \sqcap  \NNF(D) \\
 \NNF(\neg(C \sqcup D)) = \NNF(\neg C) \sqcap  \NNF(\neg D) \\
 \NNF(\neg(C \sqcap D)) = \NNF(\neg C) \sqcup \NNF(\neg D) \\ 
 \NNF(\forall R.C) = \forall R.\NNF(C) \\
 \NNF(\exists R.C) = \exists R.\NNF(C) \\
 \NNF(\neg \forall R.C) = \exists R.\NNF(\neg C) \\
 \NNF(\neg \exists R.C) = \forall R.\NNF(\neg C) \\\\
 \NNF({\cal T}) =  
\bigcup_{C \sqsubseteq D \in {\cal T} } \NNF(\neg C \sqcup  D)  \\
\NNF({\cal A}) = \bigcup_{C(a)  \in {\cal A} } \NNF(C)(a) \cup 
\bigcup_{R(a,b)  \in {\cal A} } R(a,b) \cup \\
\ \ \ \ \ \ \ \ \ \ \ \ \ \ \ \ \ \ \ \bigcup_{a = b \in {\cal A} } a = b \cup
\bigcup_{a \not= b \in {\cal A} } a \not= b 
 \end{array}
 \] 
\caption{Negation Normal Form of a Concept, a TBox and an ABox} 
\label{figure:fnn}
 \end{figure}

\begin{definition}[Isomorphism between interpretations of
$\ALC$]\\
An isomorphism between two 
interpretations $\interp $ and $\interp'$ of $\ALC$
is a bijective function $f: \Delta \rightarrow \Delta'$
such that
\begin{itemize}
\item
$f(a^{\interp}) = a^{\interp'}$
\item
$x \in A^{\interp}$ if and only if $f(x) \in A^{\interp'}$

\item
$(x,y) \in R^{\interp}$ if and only if
$(f(x), f(y)) \in R^{\interp'}$.
 
\end{itemize}

\end{definition}

\begin{lemma}
\label{lemma:isomorphism}
Let $\interp$ and $\interp'$ be two isomorphic interpretations of $\ALC$.
Then, 
$\interp$ is a model of $(\mathcal{T}, \mathcal{A})$
if and only if
$\interp'$ is a model of $(\mathcal{T}, \mathcal{A})$.
\end{lemma}
To prove the previous lemma is enough to show that $x \in C^{\interp}$ if and only if
$f(x) \in C^{\interp'}$ by induction on $C$.

\begin{theorem}[Complexity of $\ALC$]
\label{theorem:complexityALC}
Consistency of a (general) knowledge base in  $\ALC$ is ExpTime-complete.
\end{theorem} 
 
\noindent The hardness result was proved by Schild~\cite{DBLP:conf/ijcai/Schild91}.
A matching upper bound for $\ALC$ was given by De Giacomo and Lenzerini by a
reduction to PDL \cite{DBLP:conf/dlog/GiacomoL96}.
The classic Tableau algorithm is not optimal and cannot be used in this proof because
it is  NExpTime~\cite{DBLP:conf/dlog/2003handbook}.
ExpTime Tableau algorithms for checking satisfiability w.r.t.
a general Tbox are shown by Lenzerini, 
Donini and Masacci \cite{DBLP:conf/dlog/GiacomoDM96,DBLP:journals/ai/DoniniM00}.
These algorithms globally cache only  unsatisfiable sets.
Gor\'e and Nguyen show an ExpTime Tableau algorithm for checking satisfiability
of a concept in $\ALC$ w.r.t. a general Tbox that can globally cache
satisfiable and unsatisfiable sets \cite{DBLP:journals/jar/GoreN13}.
Nguyen and Szalas extend this same algorithm for checking consistency
of a knowledge base (including  a Tbox and an Abox) in $\ALC$
\cite{Nguyen2009}.

\section{Well-founded Sets and Well-founded Relations}
\label{Section:WellfoundedSets}

In this section we recall some basic notions on 
well-founded sets and relations~\cite{Winskel2010}.
In particular, the induction and recursion principles
are important for us since we will use them
in the proof of completeness of the Tableau Calculus for
$\ALCM$.

\begin{definition}[Well-founded Relation]
\label{definition:wellFounded} 
Let $X$ be a set and $\rel$ a binary relation on $X$.

\begin{enumerate}
\item Let $Y \subseteq X$.
We say that $m \in Y$ is a {\em minimal element} of $Y$
if there is no $y \in Y$ such that $y  \rel   m$.

\item 
We say that $\rel$ is {\em well-founded} (on $X$) if 
for all $Y \not = \emptyset$  such  that $Y \subseteq X$, we have that $Y$ has a minimal element.

\end{enumerate}
\end{definition}
\noindent
Note that in the general definition above the relation
$\rel$ does not need to be transitive.

\begin{lemma}  
\label{lemma:decreasingsequences}
The order $\rel$ is  well-founded on $X$  iff there are no 
infinite $\rel$-decreasing sequences, i.e., there is no
$\langle x_i \rangle_{i \in \mathbb{N} }$
such that $x_{i+1}  \rel  x_i$ and $x_i \in X$ for all
 $i \in \mathbb{N}$.
 
\end{lemma}
\noindent
The proof of the above lemma can be found in \cite{Winskel2010}.

\begin{definition}[Well-founded Set]
A set $X$ is {\em well-founded} if the set membership relation 
$\in $ is
well-founded on the set $X$.
\end{definition}

\noindent
As a consequence of Lemma \ref{lemma:decreasingsequences}, we also 
have that:

\begin{enumerate}
\item  If  $X$ is  a well-founded set then $X \not \in X$. 

\item 
If  $X$ is a well-founded set then
it cannot contain an
infinite $\in$-decreasing sequence, i.e., there is no
$\langle x_n \rangle_{n \in \mathbb{N} }$
such that $x_{n+1} \in  x_n$ and $x_n \in X$ for all
 $n \in \mathbb{N}$. 
 \end{enumerate}
 
\noindent 
An important reason that well-founded relations are interesting is because
we can apply the induction and recursion principles, e.g.,~\cite{Winskel2010}.
In this paper  both principles will be used to prove 
 correctness of the Tableau calculus for $\ALCM$.

\begin{definition}[Induction Principle]
\label{definition:inductionprinciple}
If $\rel$  is a well-founded relation on $X$, 
$\varphi$ is some property of elements of $X$, and we want to show that
$\varphi(x)$  holds for all elements $x \in X$,
it suffices to show that:\\
\noindent  
   {\em  if $x \in X$ and $\varphi(y)$ is true for all $y \in X$ such that 
    $y \rel x$, then $\varphi(x)$ must also be true. 
   } 
\end{definition}

\begin{definition}[Function Restriction]
The restriction of a function $f : X \rightarrow Y$ to a subset $X'$
of $X$ is denoted as $\restfun{f}{X'}$ and defined as follows.
\[ 
\restfun{f}{ X'} = \{ (x, f(x)) \mid  x \in X' \}
\]
\end{definition}

On par with induction,
 well-founded relations also support construction of objects by  recursion. 

\begin{definition}[Recursion Principle] 
\label{definition:recursionprinciple}
 If $\rel $ is a well-founded relation on $X$ and $F$
  a function that assigns an object $F(x, g)$
   to each pair of an element $x \in X$
    and a function $g$ on the initial segment 
    $\{y \in X \mid  y \rel  x\}$ of $X$. 
    Then there is a unique function $G$ such that for every $x \in X$,
\[
 G(x)=F(x,\restfun{G}{{\{y \in X \mid  y \rel x\}}}) 
    \]
    \end{definition}

\section{The Description Logic $\ALCM$}
\label{Section:DLALCM}

In this section, we extend the description logic $\mathcal{ALC}$ with
the    meta-modelling defined  by Motz, Rohrer and Severi~\cite{Motz2015,DBLP:conf/jist/MotzRS14}.

\begin{definition}[Meta-modelling axiom] 
A {\em meta-modelling axiom} is a  statement of the form 
 $a \eqm A$  where $a$ is an individual and $A$ is an atomic concept.
 We pronounce $a \eqm A$ as {\em $a$ corresponds to $A$ through meta-modelling}. 
 An {\em Mbox} $\Mb$  is a finite set of meta-modelling axioms.   
\end{definition}
\begin{figure}
\centering
\includegraphics[scale = 0.5]{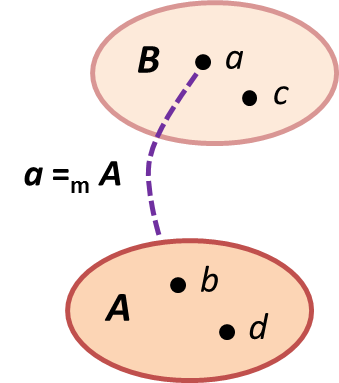}
\caption{Meta-modelling Axiom Example}
\label{fig:metaAxiom}
\end{figure}
\noindent 

\noindent 
In Figure \ref{fig:metaAxiom}, the  meta-modelling axiom 
$a \eqm A$ express that the individual $a$ corresponds to the concept $A$ through meta-modelling. 
\\
We define $\ALCM$ by keeping the same syntax for concept expressions as for $\ALC$.
\\
A {\em knowledge base $\Kb$ in ${\ALCM}$} is a triple 
$\KB{\Tb}{\Ab}{\Mb}$ where  
$\Tb$,   $\Ab$ and $\Mb$  
are  a Tbox,  Abox and  an  Mbox respectively.
The set of all individuals in $\Mb$ 
 is denoted by $\dom (\mathcal M)$  and the set of all concepts 
  by $\range (\mathcal M)$.
\\
Figure \ref{figure:boxesforfirstview} shows the Tbox, Abox and Mbox of the
knowledge base that corresponds to Figure \ref{fig:firstView}.

\begin{figure}

{\small
\begin{center}
\begin{tabular}{cc}

\begin{tabular}{l}
\textbf{Tbox} \\
$River \sqcap Lake \sqsubseteq \bot $ \\\\
\textbf {Mbox} \\

$river \eqm  River$ \\
$lake \eqm Lake$
 
\end{tabular}

&

\begin{tabular}{l}
\textbf {Abox} \\

$HydrographicObject(river)$ \\ 
$HydrographicObject(lake)$ \\
$River(queguay)$ \\ 
$River(santaLucia)$ \\
$Lake(deRocha)$ \\
$Lake(delSauce)$ 

\end{tabular}

\end{tabular}
\end{center}
}

\caption{Tbox, Abox and Mbox for Figure \ref{fig:firstView}}
\label{figure:boxesforfirstview} 

\end{figure}

\begin{definition}[Satisfiability of meta-modelling]
\label{definition:satisMeta}
An interpretation $\interp$ {\em  satisfies} (or it is a {\em model of}) $a \eqm A$
if $a^{\interp} = A^{\interp}$. 
An interpretation $\interp$ {\em satisfies} (or it is {\em  model of}) $\Mb$, denoted by $\interp \models \Mb$, if it satisfies  each statement in $\Mb$.
\end{definition}

\noindent
The semantics of $\ALCM$ makes  use of  the structured  domain elements.   
  In order to give semantics to  meta-modelling, the domain has to consists of 
  basic objects, sets of objects, sets of sets of objects and so on.

\begin{definition}[$S_n$ for $n \in \mathbb{N}$]
\label{definition:domainSet}
Given a non empty set $S_{0}$ of atomic objects, we define $S_{n}$ by induction on $\mathbb{N}$ as follows:
$S_{n+1} = S_{n} \cup \mathcal{P}(S_{n}) $
\end{definition}
It is easy to prove that $S_n \subseteq S_{n+1}$ 
and that $S_n$ is well-founded for all $n \in \mathbb{N}$.  \\ 
A set  $X \subseteq S_{n}$ can contain elements $x$ such that $x \in S_{i}$
 for any $i \leq n$.

\begin{definition}[Model of a Knowledge Base in $\ALCM$]
\label{definition:modelALCM}
An interpretation ${\interp}$ is a \emph{model of a knowledge base
$\Kb = \KB{\Tb}{\Ab}{\Mb}$ in $\ALCM$ } 
(denoted as $\interp \models \Kb$)
 if the following holds:

\begin{enumerate}

\item the domain  $\Delta$  of the interpretation is  a {\em subset}  of some $ S_{n}$ for some  $n \in \mathbb{N}$.

\item  ${\interp}$ is a model of  $\KbALC{\Tb}{\Ab}$ in $\ALC$.

\item ${\interp}$ is a model of $\Mb$.
\end{enumerate}
 
\end{definition}
\noindent
In the first part of Definition \ref{definition:modelALCM} 
we restrict the 
domain of an interpretation in $\ALCM$
 to be a {\em subset} of $S_n$.
The domain $\Delta$ can now contain  sets  since 
 the set $S_n$  is defined recursively using
the powerset operation. 
In the presence of meta-modelling, the domain $\Delta$
cannot longer consist of  only  basic objects 
and cannot be an arbitrary set either.
We require that
the domain be a  well-founded set. The reason for this is explained as follows. Suppose we have a domain  
$\Delta^{\interp} = \{X\}$ where $X = \{X\}$ is a set
that belongs to itself. Intuitively, $X$ is the set 
\[
\{\{\{ \ldots \}\}\}
\]
Clearly, a set like $X$ should be excluded from our interpretation domain since it cannot represent any
 real object from  our  usual applications in Semantic Web
(in other areas or aspects of Computer Science, representing such objects
is useful~\cite{DBLP:journals/jolli/Akman97}). 
\\
Note that $S_0$ does not have to be the same for all models of a knowledge base. 
\\
The second part of Definition \ref{definition:modelALCM} refers to the $\ALC$-knowledge base
 without the Mbox axioms. In the third part of the definition, we add another condition that the model must satisfy considering the meta-modelling axioms. This condition restricts the interpretation of an individual that has a corresponding concept through meta-modelling  to be equal to the concept interpretation. 

\begin{example}
We define a model for the knowledge base of Figure \ref{figure:boxesforfirstview} where 
\[ S_0 = \{ queguay, santaLucia, deRocha, delSauce \}\]
\[
\begin{array}{ll}
 \Delta = \{ &  queguay, santaLucia, deRocha, delSauce,  \\
 &\{queguay, santaLucia\}, \\
 &\{deRocha, delSauce\}, \\
 &\{\{queguay, santaLucia\}, \{deRocha, delSauce\} \}\\
 ~ ~ ~ ~ ~ ~ ~ ~ ~ ~  \}\\
\end{array}
\]


 The  interpretation is defined on 
the individuals with meta-modelling and the
 corresponding atomic concepts to which they are equated  as follows:
\begin{center}
\begin{tabular}{ll}
$river^{\interp} = River^{\interp} = \{queguay, santaLucia\}$  \\
$lake^{\interp} = Lake^{\interp} = \{deRocha, delSauce\}$  
\end{tabular}
\end{center} 
and on the remaining atomic concept which does not appear on
the MBox the interpretation is
defined as follows:
\[
\begin{array}{l}
HydrographicObject^{\interp} 
\\
 = \{ river^{\interp}, lake^{\interp} \}\\
 = \{\{queguay, santaLucia\}, \{deRocha, delSauce\} \}
\end{array}
\]

\end{example}

\begin{definition}[Consistency  in $\ALCM$]
\label{definition:consistentbaseALCM}
We say that a knowledge base $\Kb =\KB{\Tb}{\Ab}{\Mb}$ is \emph{consistent}  (\emph{satisfiable})
if there exists a model of $\Kb$.
\end{definition}

\begin{definition}[Logical Consequence in $\ALCM$]
We say that $\statement$ is a \emph{logical consequence of} 
$\Kb = \KB{\Tb}{\Ab}{\Mb}$ 
(denoted as $\Kb \models \statement$) if all models of $\Kb$ are also models  
of $\statement$ where $\statement$ 
is any of the following  
$\ALCM$ statements, i.e.,
\begin{center}
 $C \sqsubseteq D$  \ \ \ $C(a)$  
 \ \ \ $R(a, b)$  \ \ \  $a \eqm A$   \ \ \ $a = b$  \ \ \ $a \not = b$. 
 \end{center}
  
\end{definition}

\begin{definition}[Meta-concept]
\label{definition:metaconcept}
We say that $C$ is a meta-concept in $\Kb$ if there exists an individual $a$ such that $\Kb \models  C(a)$ and $\Kb \models a \eqm A$.
\end{definition}
\noindent
Then, $C$ is a meta-meta-concept if there exists an individual $a$ such that 
$\Kb \models  C(a)$, $\Kb \models a \eqm A$ and $A$ is a meta-concept.
Note that a meta-meta-concept is also a meta-concept. 
\\
\\
\noindent
We have some new inference problems:

\begin{enumerate}

\item {\em Meta-modelling}. Find out whether $a \eqm A$ or not.

\item {\em Meta-concept}. Find out whether $C$ is a meta-concept or not.

\end{enumerate}
\noindent
As most inference problems in Description Logic the  above  two problems can be reduced to satisfiability.
For the first problem, note that  since  $a \not =_m A$ is not directly available in the 
 syntax, we have replaced it by $a \not = b $  and $ b \eqm A$
 which is  an equivalent statement 
to the negation of  $a \eqm A$. 
For the proof of the following two lemmas, see \cite{Motz2015}.

\begin{lemma}
\label{lemma:firstreasoningproblem}
$\Kb \models a \eqm A$ if and only if 
for some new individual $b$, $\Kb \cup \{a \not = b, b \eqm A \}$ is  inconsistent.
\end{lemma}

\begin{lemma}
\label{lemma:secondreasoningproblem}
A concept $C$ is a meta-concept
if and only if
for some  individual $a$ we have that
 $\Kb \cup \{\neg C(a) \}$  is  inconsistent 
and
for some new individual $b$, $\Kb \cup \{a \not = b, b \eqm A \}$ is inconsistent.
\end{lemma}

\noindent
Next lemma explains more formally
 why in our definition of ABox
  we included expressions of the form $a = b$ and 
 $a \not = b$. 
If we have an equality  $A \equiv B$ between concepts
then  $a$ and $b$ should be equal. Similarly, if we have that $A$ and $B$ are different 
then $a$ and $b$ should be different.
In other words, since we can express equality and difference between concepts, we also
need to be able to express equality and difference at the level of individuals.

\begin{lemma}[Equality Transference] \label{lemma:equalitytransference} \mbox{  }\\
Let $\Kb = \KB{\Tb}{\Ab}{\Mb}$ be 
a knowledge base,
$\Kb \models a \eqm A$ and $\Kb \models b \eqm B$.

\begin{enumerate}

\item If $\Kb \models a = b$ then $\Kb \models A \equiv B$.

\item If $\Kb \models A \equiv B$ then $\Kb \models a = b$.

\end{enumerate}

\end{lemma}
 The proof of the above lemma is immediate since 
$a$, $b$, $A$ and $B$  are all interpreted
as the same object.
\\
\\
 We  define a tableau algorithm for checking consistency 
of a knowledge base in $\ALCM$
by extending the standard tableau algorithm for $\ALC$.
The expansion rules for $\ALCM$  consist of the rules for
$\ALC$ 
and some additional  expansion rules 
  for meta-modelling (see Figure \ref{figure:ALCMexpansionrules}). 
 The additional 
 expansion rules 
 deal with 
  the equalities and inequalities between individuals with meta-modelling  
 which need to be transferred to the level of concepts as equalities and inequalities between the corresponding concepts. 
 We also need to add 
 an extra  condition  that checks  for circularities (with respect to membership) avoiding non well-founded sets. 
\\
\\
\noindent
 A  \emph{completion forest}  $\forest$ 
for  an  $\ALC$ knowledge base consists of
\begin{enumerate}
\item a set of nodes, labelled with individual names or variable names  (fresh individuals which do not belong to the ABox),  

\item directed edges between some pairs of nodes,
\item for each node labelled $x$, a set $\forest(x)$ of concept expressions, 
\item for each pair of nodes $x$ and $y$, a set $\forest(x, y)$ containing role names,
and
\item two relations between nodes, denoted by $\eql$ and $\neql$.
These relations keep record of the equalities and inequalities 
 of nodes  in the algorithm. 
  The relation $\eql$ is assumed to be 
   reflexive, 
   symmetric and transitive while $\neql$ is assumed to be symmetric.   
We also assume that the relation $\neql$ is {\em compatible with $\eql$}, i.e., 
if  $x' \eql x$ and $x \neql y$ then  $x' \neql y$ for all $x,x', y$.
In the algorithm, every time we add a pair in $\eql$, we close $\eql$ under
reflexivity, 
 symmetry and transitivity. Moreover, every time we add a pair in either $\neql$  or $\eql$,
  we close $\neql$ under compatibility with $\eql$. 

\end{enumerate}

\noindent
We assume that $\Tb$ y $\Ab$ have already been converted into negation normal form.
 
\begin{definition}[Initialization] \label{definition:initializationSHIQM}
The \emph{initial  completion forest} for ${\Kb} = \KB{\Tb}{\Ab}{\Mb}$ 
 is defined by the following procedure.

\begin{enumerate}
\item 
For each  individual $a$ in the knowlegde ($a \in {\As} \cup {\Mb}$)
set $a \eql a$.
\item 
For each $ a = b  \in \As$, set $a \eql b$. 
We also choose an individual as a representative of 
each equivalence class. 
\item For each  $a \not = b$ in $\As$, set $a \neql b$.
\item  
For each  $a \in {\As} \cup {\Mb}$, we do the following:
\begin{enumerate}
\item 
in case $a$ is a representative of an equivalence class then
set 
$\forest(a) = \{C \mid C(a')\in \As, a \eql a' \}$;
\item in case $a$  is not a representative of an equivalence class
then set $\forest(a) = \emptyset$. 
\end{enumerate}
 \item  
For all  $a, b \in {\As}\cup {\Mb}$ that are representatives of some equivalence
class, 
if $\{ R \mid R(a',b') \in \As, a \eql a', b \eql b'\} \not = \emptyset$
then create an edge from $a$ to $b$ and 
set $\forest(a,b) = \{ R \mid R(a',b') \in \As, a \eql a', b \eql b'\}$.
\end{enumerate}
\end{definition} 
 \noindent
Note that in case $a$  is not a representative of an equivalence class
and it has some axiom $C(a)$,
we set $\forest(a) = \emptyset$ because we do not want to apply any expansion rule to
$\forest(a)$. The expansion rules will only be applied to the
representative of the equivalence class of $a$.
\begin{definition}[Contradiction]
 A completion forest $\forest$   has a {\em contradiction} if either \begin{itemize}
\item $A$ and $\neg A$ belongs to  $\forest(x)$  for some atomic concept $A$ and node $x$ or 
\item there are nodes $x$ and $y$ such that $x \neql y$ and $x \eql y$.
\end{itemize}
\end{definition}
\noindent

\noindent We say that  a node $y$ is  a {\em successor} of a node $x$ if  $\lab(x,y)  \not = \emptyset$.
We define that $y$ is a {\em descendant} of $x$ by induction.
\begin{enumerate}
\item 
Every successor  of $x$,  which is a variable, is a descendant  of $x$. 
\item 
Every successor  of a descendant of $x$, which is  a variable, is also a descendant of $x$.
\end{enumerate}

\begin{definition}[Blocking]
We define the notion of blocking by induction.
A node $x$ is  blocked by a node $y$ if $x$ is a descendant of $y$  and $\lab(x) \subseteq \lab(y)$
or $x$ is a descendant of $z$ and $z$ is blocked by $y$. 
\end{definition}

\begin{definition}[$\ALCM$-Complete]
 A  completion  forest $\forest$  is {\em $\ALCM$-complete} (or just {\em complete}) 
 if none of the rules of  Figure \ref{figure:ALCMexpansionrules}  is applicable.
\end{definition}

\begin{definition}[Circularity]
\label{definition:cycle}
We say that the  completion forest 
 $\forest$  has a circularity 
 with respect to $\Mb$ if there is a sequence of meta-modelling axioms  $ a_0 \eqm A_0$, $ a_1  \eqm A_1 $, $ \ldots$ $ a_n \eqm  A_n$  all in $\mathcal{M}$ such that 
\[\begin{array}{ll}
A_1 \in \forest(x_0) & x_0 \eql a_0\\
A_2 \in \forest(x_1) & x_1 \eql a_1\\
\vdots & \vdots \\
A_n \in \forest(x_{n-1}) \ \ \ & x_{n-1} \eql a_{n-1}\\
A_0 \in \forest(x_{n}) & x_{n} \eql a_{n}
\end{array}
\]  
\end{definition}

\noindent
After initialization, the tableau algorithm proceeds by non-deterministically applying the 
{\em expansion rules for $\ALCM$} given 
in Figure~\ref{figure:ALCMexpansionrules}.  \\
 The algorithm says that the ontology $\KB{\Tb}{\Ab}{\Mb}$ 
 is consistent iff 
the expansion rules can be applied in such a way they yield a complete forest
$\forest$  without contradictions nor circularities.   Otherwise the algorithm 
says that it is inconsistent. 
Note that due to the non-determinism of the algorithm, implementations of it have to guess the choices and possibly have to backtrack to choice points if a choice already made has led to a contradiction. 
The algorithm stops when we reach {\em some} $\forest$ 
 that is  complete, has neither contradictions nor circularities or 
 when all the choices have yield a  forest with contradictions or circularities.

\begin{figure}[H]

\begin{description}

\addtolength{\itemsep}{-4pt}

\item 
{\bf $\sqcap$-rule}: \\
If $C \sqcap D \in \forest(x)$ and $\{C, D \} \not \subseteq \forest(x)$ then set $\forest(x) \leftarrow \{C, D \}$.

\item 
{\bf $\sqcup$-rule}: \\
If $C \sqcup D \in \forest(x)$ and $\{C, D \} \cap \forest(x)= \emptyset$ then set $\forest(x) \leftarrow \{C \}$ or                      $\forest(x) \leftarrow \{D \}$.

\item 
{\bf $\exists$-rule}: \\
If $x$ is not blocked, $\exists R.C\in \forest(x)$ and there is no $y$ with $R \in \forest(x,y)$ and $C \in \forest(y)$ then 

1. Add a new node with label $y$ (where $y$ is a new node label), 
2. set $\forest(x,y) = \{R\}$,\\
3. set $\forest(y) = \{ C \}$. 

\item 
{\bf $\forall$-rule}: \\
If $\forall R.C\in \forest(x)$ and there is a node $y$ with $R \in \forest(x,y)$ and $C \not \in \forest(y)$ then set $\forest(x) \leftarrow C$.

\item 
{\bf ${\cal T}$-rule}: \\
If $C \in {\cal T}$ and $C \not \in \forest(x)$, then $\forest(x) \leftarrow C$.

\item 
{\bf $\eql$-rule}: \\
Let $a \eqm A$ and $b \eqm B$ in $\mathcal{M}$.
If $a \eql b$ 
and $A \sqcup \neg B, B \sqcup \neg A$ does not belong to $\mathcal{T}$ then 
add $ A \sqcup \neg B, B \sqcup \neg A$ 
to
${\cal T}$.
\item 
{\bf $\neql$-rule}: \\ 
Let $a \eqm A$ and $b \eqm B$ in $\mathcal{M}$.
If $a \neql b$ 
and there is no node 
 $z$ such that 
$A \sqcap \neg B \sqcup B \sqcap \neg A \in \forest(z)$  
then  
create a new  node $z$ 
 with 
$\forest(z) = \{ A \sqcap \neg B \sqcup B \sqcap \neg A \}$ 

\item 
{\bf close-rule}: \\ 
Let $a \eqm A$ and $b \eqm B$ in $\mathcal{M}$
 where  $a \eql x$, $b \eql y$, 
  $x$ and $y$ are their respective representatives of the equivalence classes.
If  neither $x \eql  y$ nor $x  \neql y$  
then we add either $x \eql y$ or $x \neql y$.
In the case $x \eql y$, we also do the following: \\
1. add $\forest(y)$ to $\forest(x)$, \\
2. for all directed edges from $y$ to some $w$,
create an edge from $x$ to $w$ if it does not exist with $\forest(x,w) = \emptyset$,\\
3.  add $\forest(y, w)$ to $\forest(x, w)$, \\
4.  for all directed edges from some $w$ to $y$,
 create an edge from $w$ to $x$ if it does not exist with $\forest(w,x) = \emptyset$,\\
5. add  $\forest(w, y)$ to  $\forest(w, x)$, \\
6. set  $\forest(y) = \emptyset$ and remove all edges from/to $y$.

\end{description}

\caption{Expansion Rules for $\ALCM$}
\label{figure:ALCMexpansionrules}
\end{figure}

This tableau algorithm has complexity NEXP
and hence, it is not optimal.
In the following section 
 we define an optimal algorithm for
$\ALCM$ and prove that is ExpTime.

\section{A Tableau Calculus for $\ALCM$}
\label{Section:TableauCalculus}

In order to eliminate non-determinism, the
completion trees of the standard tableau algorithm 
are replaced with structures called {\em and-or graphs} \cite{Nguyen2009}. 
In such structure, both branches of a non-deterministic choice introduced by 
disjunction are explicitly represented. 
Satisfiability of the branches is propagated bottom-up and if it reaches
an initial node, we can be sure that a model exists. 
A global catching of nodos and a proper rule-application strategy is used to guarantee the exponential bound on the size of the graph.
The and-or graphs are built using a Tableau Calculus. 
  Our Tableau Calculus for $\ALCM$ is an extension of the Tableau Calculus for $\ALC$
  given by Nguyen and Szalas \cite{Nguyen2009},
  which adds  new rules for handling meta-modelling. 
\\
From now on,  we assume that  Tboxes, Aboxes and all concepts are in negation normal form  (NNF)
(see Figure \ref{figure:fnn}). 
 Note that the concepts in an Mbox are already in NNF since they are all atomic. 
\\
The Tableau Calculus for $\ALCM$ is defined  by the tableau rules of 
Figures~\ref{figure:plainrules} and \ref{figure:primerules}.
In the premises and conclusions of these rules, we find \emph{judgements}
$\judgement$  that
are built using the following grammar.

\begin{enumerate}
\item {\em Simple  judgements} $\judgement$ are of the
following three forms:
\[
\begin{array}{ll}
\derboxes{\Tb}{\Ab}{\Mb} & \ \ \  \mbox{{\em  \basejuicio}}  \\\\
\derconcepts{\Tb}{\Xs}  & \ \ \  \mbox{{\em  \varjuicio}}  \\\\
\derbot & \ \ \  \mbox{{\em absurdity judgement}}    
\end{array}
\]

\item  {\em Or-judgements} are of the form
 $\judgement_1 \ornode \judgement_2$,  where  $\judgement_1$ and $\judgement_2$ are simple judgements.   

\item {\em And-judgements} are of the form 
$\judgement_1 \andnode  \ldots \andnode \judgement_k$, where  $\judgement_i$  are simple 
judgements for all $1 \leq i \leq k$. 

\end{enumerate}
\noindent
The  first basic judgement is  just the
knowledge base $\KB{\Tb}{\Ab}{\Mb}$. 
 After initialization, the Aboxes $\Ab$ of the  Tableau  Calculus 
  do not contain equality assertions $ a =b$
because we choose the individual $a$ as a
 canonical representative and replace $b$ by $a$.
  We write $\As [a/b]$ and $\Mb [a/b]$ to denote
  the replacement of $b$ by $a$ in $\As$ and $\Mb$ respectively.
  \\
For the second basic judgement, we have that
 $\Tb$ is a Tbox and $\Xs$ is a set of concepts.
The set $\Xs$ intuitively represents the set of concepts
satisfiable by a certain unknown $x$.
The key feature of this tableau calculus is actually not to create  variables
$x$ for the existential restrictions as the standard tableau for Description Logic does.
It actually forgets the variables
and only cares about the set $\Xs$ of concepts that should be satisfiable.
By not  writing $x$ explicitly, the set $\Xs$ can be shared
by another  existential coming from another branch
in the graph 
and this will allow us to obtain an 
ExpTime algorithm for checking consistency.
\\
Recall from Definition \ref{definition:consistentbaseALCM} and \ref{definition:modelALCM}
that 
$\derboxes{\Tb}{\Ab}{\Mb}$
is satisfiable if there exists $\interp$ such that
\begin{enumerate}
\item $\interp \models  \bigsqcap \Tb \equiv  \top$ thus $\interp$ is a model of 
$\Tb$, 
\item $\interp \models \Ab$ and 
\item $\interp \models \Mb$.
\end{enumerate}
\noindent
Recall  that
$\derconcepts{\Tb}{\Xs}$
is satisfiable    
if there exists $\interp$ such that
\begin{enumerate}
\item 
$\interp \models  \bigsqcap \Tb \equiv  \top$ thus $\interp$ is a model of 
$\Tb$ and 
\item $\interp$ satisfies the set $\Xs$ of concepts, i.e. 
$(\bigsqcap \Xs)^\interp \not = \emptyset$.
\end{enumerate}  
%
We have that $\derbot$ is always unsatisfiable.
\\
We say that 
$\judgement_1 \andnode \ldots \andnode  \judgement_k$
is satisfiable if 
$\judgement_i$ is satisfiable for all $1 \leq i \leq k$.
We use
$\bigandnode  \{\judgement_1, \ldots, \judgement_k\}$
to denote
$\judgement_1 \andnode \ldots \andnode \judgement_k$.
\\
We say that
$\judgement_1 \ornode  \judgement_2$
is satisfiable if 
 $\judgement_1$ or $\judgement_2$ is satisfiable. 
 \\  
The tableau rules are written downwards, they have only one simple judgement as premise
and can have a sequence of simple judgements as conclusion.
{\em Unary rules} have only a simple  judgement  as conclusion
 and they are of the form:
\[
 \inference{\judgement_0}{\judgement_1}
\]
where $\judgement_0$ and $\judgement_1$ are both simple.
\\
{\em Non-unary rules} have a conclusion which is a and-judgements or an or-judgements.
 In the first case, we have an {\em and}-rule which
is of the form:
\[
 \inference{\judgement_0}{\judgement_1 \andnode \ldots \andnode \judgement_k}
\]
where $\judgement_i$ is simple for all $1 \leq i \leq k$.
In this case, the arity of the rule is $k$. 
\\
When the conclusion is an or-judgements, we have
an {\em or}-rule which is of the form:
\[ \inference{\judgement_0}{\judgement_1 \ornode \judgement_2}\]
where $\judgement_1$ and $\judgement_2$ are   simple. 
In this case, the arity of the rule is $2$.
\\
%
A \emph{bottom rule} is a rule whose conclusion is $\derbot$.
 The $\transR$ and $\transPR$ rules are called \emph{transitional}, the rest of the 
rules are called  \emph{static}.


\noindent The rules $\orR$, $\orPR$ and $\closePR$ are all or-rules
 and all of them have arity $2$.
The rules $\transR$ and $\transPR$ are and-rules. The arity of the and-rules
depend on the number of existentials in $\Ab$ or in $\Xs$.  
The rest of the rules are unary since their conclusion is a simple judgement.\\
In the Tableau Calculus, the rules $\orR$, $\orPR$ and $\closePR$ are deterministic.
The inherent choice  of these rules
is explicitly represented using the or-judgement. 
\\
The Tableau Calculus does not have an explicit Tbox-rule
as the standard tableau algorithm. Instead,
 this rule is ``spread inside other rules'', namely $\transR$, $\transPR$ and $\difference$-rules
 ensuring  that the new individuals satisfy the concepts in the Tbox.
Also, the initialization  ensures that
 all the individuals of the initial knowledge base
satisfy the concepts of the Tbox (see Definition \ref{definition:andorgraph}).\\
The rules $\transR$ and $\transPR$ need to be and-rules because
all the outermost existentials are treated simultaneously~\footnote{The concept $\exists R. \exists S. C$
has two nested existentials and the outermost one is $\exists R$.}.
Each of these rules is a compact way of doing 
what was necessary to do with  three rules, 
$\exists$-rule, $\forall$-rule and the Tbox-rule in 
 the standard tableau algorithm.

The rules $\transR$ and $\transPR$ are indirectly introducing unknown individuals
 for each outermost existential. Those unknown individuals should  satisfy all concepts
in the Tbox and inherit all concepts from the  $\forall$
of the corresponding role.\\
%
%

We explain the intuition behind the new rules that deal with meta-modelling,
which are $\botThreePR$,
$\closePR$, $\equal$ and $\difference$. 
If $a \eqm A$ and $b \eqm B$ then the individuals $a$ and $b$ 
represent concepts. Any equality at the level of 
individuals should be transferred as an equality between
concepts and similarly with the difference. 
\\ 
 Note that the Aboxes of the Tableau  Calculus 
 do not contain  equality assertions $ a =b$
because  we choose the individual $a$ as a
 canonical representative and replace $b$ by $a$ in the initialization as well as
 in the $\closePR$-rule.
\\
The $\closePR$-rule is an or-rule that 
distinguishes between the two possibilities: either   $a$ and $ b$ are equal  or
they are different.
In the case   $a$ and $b$ are equal, we choose $a$ as canonical representative
and replace $b$ by $a$. Note that $b$ is replaced by $a$ in the Mbox
and we get $a \eqm B$ as the result of replacing $b$ by $a$ in $b \eqm B$. 
In the case $a$ and $b$ are different, we simply add the axiom $a \not = b$ to the Abox.

The $\equal$-rule transfers an ``equality between individuals'' to 
the level of concepts. Since we choose canonical representatives for individuals,
this happens only when $a \eqm A$ and $a \eqm B$. 
Instead of adding    the statement  $A \equiv B$ to the TBox, 
we add its negation normal form which is
$(A \sqcap \neg B) \sqcup (\neg A \sqcap B)$.
Since we add this new concept to the TBox, we 
also have to add that all the existing individuals
satisfy this concept.
Note that the principal premise $a \eqm B$  
is removed from the Mbox and does
not appear in the conclusions of the rule.

The $\difference$-rule transfers the difference between individuals 
to the level of concepts. If  $a \not =  b$ is in the Abox, 
then we should add that 
$A \not \equiv B$. However, we cannot add
 $A \not \equiv B$ because the negation of $\equiv$ is not
directly  available in
the language. So, what we do is to 
 replace it by an equivalent statement,  i.e.  we add an element $d_0$
  that  is in $A$ but not in $B$ or it is in $B$ but it is not in $A$.
Note also that the individual $d_0$ should also satisfy all the concepts
that are in the Tbox $\Tb$.

The $\botThreePR$-rule ensures that there are no circularities and hence, 
the domain of the canonical interpretation is well-founded. This rule uses the following definition:

\begin{definition}[Circularity of an Abox w.r.t an Mbox]
\label{definition:circular}
We say that $\Ab$ has a circularity w.r.t. $\Mb$, denoted as
$\circular(\Ab, \Mb)$, 
if   there is a sequence of meta-modelling axioms 
$ a_1  \eqm A_1 $, $ a_2  \eqm A_2 $, $ \ldots$, $ a_n \eqm  A_n$ all in $\Mb$ such that $A_1(a_2)$, $A_2(a_3)$, \ldots, $A_n (a_1)$ are in $\Ab$. 
\end{definition}
\noindent
For example, the Abox $\Ab = \{A(a), B(b) \}$ has a circularity w.r.t to the
Mbox $\Mb = \{ a \eqm A, b \eqm B\}$.
\\
A {\em formula} is either  a concept, or an Abox statement or an Mbox statement. 
Formulas are denoted by greek letters $ \varphi$, etc.  
The distinguished formulas of the premise 
are called the \emph{principal formulas} of the rule.

\begin{remark}
\mbox{ }

\begin{itemize}

\item 
Note that the rules
$\orPR$,
$\andPR$,
$\forallPR$,   $\difference$ 
and $\closePR$ need a side condition to ensure
termination (otherwise the same rule can be vacuously applied infinite times).
The $\equal$-rule does not need a side condition because the
axiom $a \eqm B$ is removed from the Mbox.

\item 
Note that the rules $\orPR$ and 
$\andPR$ keep the principal formula while
the rules $\orR$ and $\andR$ do not.
For the rules 
$\orPR$ and 
$\andPR$, we cannot remove the principal formula
in the conclusion because 
this could lead to infinite applications of
these rules  when we combine any of these
rules with the $\forallPR$-rule. 
\end{itemize}
\end{remark}

\begin{figure}

\[
\begin{array}{lr}
\inference[$\botR$]{\derconcepts {\Ts} {\Xs \cup 
\{ A, \neg A\} }}
                          {\derbot} 
& \ \ 
\inference[$\orR$]{\derconcepts {\Ts} {\Xs \cup \{C \sqcup D \}}}
                       {\derconcepts{\Ts} {\Xs \cup \{ C \}} \ornode \derconcepts{\Ts}{\Xs \cup \{D \}}}
\\
\\ 
\inference[$\andR$]{\derconcepts{\Ts}{\Xs \cup \{ C \sqcap D \}}}
                               {\derconcepts{\Ts}{\Xs \cup \{ C, D \}}} 
& \ \
\inference[$\transR$]
             {\derconcepts{\Ts}{\Xs}}
             {\bigandnode   \{ \derconcepts{\Ts}
                                        {\Xs_{\exists R. C}}
                                         \ \suchthat \ \exists R.C \in \Xs \}}  
\\
\\
& \mbox{where }
\Xs_{\exists R.C} = \{ C \}  \cup 
 \{ D \ \suchthat   \ \forall R.D \in  \Xs \} \cup 
                                           \Ts
\end{array}
\]

\caption{Tableau rules for variable judgements  }

\label{figure:plainrules}

\end{figure}

\begin{figure}
\[
\begin{array}{l}
\begin{array}{lr}
\inference[$\botOnePR$]
          {\derboxes{\Ts}{\As \cup \{ B(a), \neg B(a) \}}{\Mb}}
          {\derbot}
&
\begin{array}{r}
\inference[$\andPR$]
{\derboxes{\Ts}{\As \cup \{(C \sqcap D)(a)\} }{\Mb}}
{\derboxes{\Ts}{\As' }{\Mb} }
{\Scale[0.8]{\begin{array}{c}
              C(a) \not \in \As \\
              \mbox{ and}\\ 
             D(a) \not \in \As  \end{array}
             }} \\[1.5em]
            \mbox{where } 
\As' = \As \cup \{ (C \sqcap D) (a), C(a), D(a)\}
\end{array}
             \\
             \\
\inference[$\botTwoPR$]
          {\derboxes{\Ts}{\As \cup \{ a \not = a \}}{\Mb}}
          {\derbot}
&
\begin{array}{r}
\inference[$\orPR$]
{\derboxes{\Ts}
         {\As \cup \{(C \sqcup D)(a)\}}{\Mb}}
{\derboxes{\Ts}{\As'}{\Mb}
 \ornode \derboxes{\Ts}{\As''}{\Mb}} 
{\Scale[0.8]{ \begin{array}{c} 
             C(a) \not \in \As \\
              \mbox{ or}\\
                   D(a)\not \in  \As  \end{array} }} 
\\[1.5em]
 \mbox{where } \As'= \As \cup \{(C \sqcup D)(a), C(a)\} 
 \\ 
\As'' = \As \cup \{ (C \sqcup D)(a),  D(a) \}
\end{array}
\\
\\
\inference[$\botThreePR$]
          {\derboxes{\Ts}{\As}{\Mb}}
          {\derbot}~{\Scale[0.8]{ \circular(\As, \Mb)}}
\end{array}
\\
\\
\inference[$\forallPR$]
           {\derboxes{\Ts}
                     {\As \cup \{ \forall R.C(a), R(a, b) \}}
                     {\Mb}}
           {\derboxes{\Ts}{\As \cup \{ \forall R.C(a), R(a, b), C(b) \} }{\Mb}}~{ \Scale[0.8]{C(b) \not \in \As }}
\\
\\
\inference[$\transPR$]
{\derboxes{\Ts}{\As}{\Mb}}
  { \bigandnode \{  
             \derconcepts{\Ts}{ \Xs_{\exists R. C(a)}} \suchthat \ \exists R.C(a) \in \As \}
  }
  \\[1.5em]
\mbox{where } \Xs_{\exists R. C(a)}  =
   \{ C \} \cup 
              \{ D \ \suchthat \ \forall R.D(a) \in  \As \} \cup 
              \Ts
  \\
  \\ 
 \inference[$\closePR$]
{\derboxes{\Ts}{\As}{\Mb}}
 {\derboxes{\Ts}{\As'}{\Mb'} \ornode \derboxes{\Ts}{\As'' }{\Mb}}
                {\Scale[0.8]{
                     \{a, b \} \subseteq \dom (\Mb) \ \ \  
                                             a \not = b \not \in \As  \ \ \
                                             a \not = b 
                              }} 
\\[1.5em]
 \mbox{where } \As' = \As [a/b], \Mb' = \Mb[a/b]
 \mbox{ and } \As'' = \As \cup \{ a \neq b \}
 \\
\\
\inference[\equal]
  { \derboxes{\Ts}{\As}{\Mb \cup \{ a \eqm A, a \eqm B \}} }
 {\derboxes{\Ts'}{\As \cup \As'}{\Mb \cup \{ a \eqm A \} }}
\\[1.5em]
\mbox{where } 
\Ts' =  \Ts \cup \{ 
 (A \sqcup \neg B),  (B \sqcup \neg A) \}
\\
\ \ \ \ \ \ \ \ \ 
\As' =  
 \{ ((A \sqcup \neg B) \sqcap (B \sqcup \neg A))(d)  
 \suchthat 
 d \in \dom (\As) \cup \dom (\Mb) \cup \{ a \} \}
\\
\\ 
\inference[$\difference$]
{\derboxes{\Ts}{\As \cup \{ a \neq b \}}
{\Mb \cup \{a \eqm A, b \eqm B \}}}
 {
 \derboxes{\Ts}{\As'}{\Mb \cup\{ a \eqm A,
  b \eqm B \}}  }~{\Scale[0.8]{{\sf cond}_{\difference}}}
\\[1.5em]
\mbox{where } 
{\sf cond_{\difference}}\mbox{ means  that
  there is no $d$ such that }
  (A \sqcap \neg B \sqcup \neg A \sqcap B)(d)  \in \As
\\
\hspace{1cm}
\mbox{{\small 
$\As' = \As \cup \{a \not = b \} \cup \As''$
}}
\\
\hspace{1cm}
\mbox{{\small 
$\As'' =  \{(A \sqcap \neg B \sqcup \neg A \sqcap B)(d_0)\}
 \cup \{ C(d_0)    \suchthat  C \in \Ts \}$}} \\
\hspace{2cm} \mbox{ {\small for a new individual $d_0$}}
\end{array}
\]

\caption{Tableau rules  for base judgements}
\label{figure:primerules}
\end{figure}

\begin{definition}[Preferences]
\label{definition:preferences}
The rules are applied in the following order:

\begin{enumerate}

\item The  bottom rules -- which  are 
$\botR$, $\botOnePR$, $\botTwoPR$ and $\botThreePR$ --
are applied with higher priority.

\item The rest of the unary rules -- which are $\andR$, $\forallPR$, $\andPR$, 
$\equal$ and $\difference$ --
are applied only if no bottom rule   is applicable.

\item The $\orR$, $\orPR$ and $\closePR$-rules
 are applied  only if  no unary rule is applicable.

\item The $\transR$ and $\transPR$-rules
are applied if no other rule is applicable.

\end{enumerate}

\end{definition}

\noindent
We now describe how to construct an {\em  and-or graph} 
 from the tableau calculus of Figures \ref{figure:plainrules} and
\ref{figure:primerules}. 
This graph is built  using {\em global caching}, i.e. the graph
contains at most one node with that judgement as label and this node is
processed (expanded) only once.
The complexity of the tableau calculus without global caching would
be double exponential. But using global caching, we do not repeat
nodes and the complexity is exponential.

\begin{definition}[And-or graph for a knowledge base in $\ALCM$]
\label{definition:andorgraph}
Let $\Kb = \KB{\Tb}{\Ab}{\Mb}$ be a knowledge base in $\ALCM$.
The {\em and-or graph for $\KB{\Tb}{\Ab}{\Mb}$}, also called
 {\em tableau for $\KB{\Tb}{\Ab}{\Mb}$}, 
is a graph  $\andorgraph$ constructed as follows.

\begin{enumerate}

\item The graph contains nodes of three kinds:  {\em \varnode s},{\em \basenode s}
and absurdity nodes. 
The label of a \varnode \ is $\derconcepts{\Ts}{\Xs}$.
The label of a \basenode \ is 
$\derboxes{\Tb}{\Ab}{\Mb}$. 
The label of an absurdity node is $\bot$.

\item 
The root of the graph is a  \basenode \  whose label is  $\derboxes{\Tb_0}{\Ab_0}{\Mb_0}$
where 
\[
\begin{array}{ll}
\Tb_0 & := \Tb \\
\Ab_0 & := \Ab^{*} \cup 
\{C(a) \ \suchthat C \in \Tb, \  a  \in \dom(\Ab^{*}) \cup \dom(\Mb^{*}) \} \\
 
\Mb_0 & := \Mb^{*} \\
\end{array}
\]

and $\Ab^{*}$ and $\Mb^{*}$ are obtained from $\Ab$ and $\Mb$
 by choosing a canonical representative $a$ for each assertion $a = b$ and replacing $b$ by $a$.

\item 
 Base nodes   are expanded using the  rules of 
Figure~\ref{figure:primerules}
while  variable nodes  are expanded using the rules of Figure~\ref{figure:plainrules}.

\item For every node $v$ of the graph, if a $k$-ary rule $\delta$ is applicable to
(the label of) $v$ in the sense that  an instance of $\delta$ has the label of $v$ as
premise and $Z_1, \ldots , Z_k$ as possible conclusions, then choose such a rule according
to the preference  of Definition \ref{definition:preferences} 
and apply it to $v$ to make $k$ successors $w_1, \ldots, w_k$
with labels  $Z_1, \ldots , Z_k$, respectively.

\item If the graph already contains a node $w_i'$ with label $Z_i$ then instead of creating a new node $w_i$ with label $Z_i$ as a successor of $v$ we just connect $v$ to $w_i'$ and assume $w_i = w_i'$.

\item If the applied rule is $\transR$ or $\transPR$ then
 we label the edge $(v, w_i)$ by 
 $\exists R. C$ if 
 the principal formula is either
  $\exists R.C$ or $(\exists R.C)(a)$.

\item If the rule applied to $v$ is an or-rule then $v$ is an \emph{or-node}.
If the rule applied to $v$ is an and-rule then
$v$ is an   \emph{and-node}.

 \item If no rule is applicable to $v$ then $v$ is an  \emph{end-node}  as well as an and-node. 

\end{enumerate}
\end{definition}

\begin{remark}
\label{remark:graph}
We make the following observations on the construction of the and-or graph.
 
\begin{itemize}

\item  The graph cannot contain edges from a  variable  node to a base node. 

\item 
Each non-end node is ``expanded'' exactly once, using only one rule.
Expansion continues until no further expansion is possible.

\item The nodes have unique labels.

\item Nodes expanded by applying a non-branching static  rule can be treated
either as or-nodes or and-nodes. We choose to treat them as or-nodes.
Applying the $\andR$ to a node causes the node to become 
an or-node (which might seem counter-intuitive).
 
\item  The graph is finite (see Lemma \ref{lemma:cantNodesGraph}).

\end{itemize}

\end{remark}

\begin{definition}[Marking] 
\label{definition:marking}
A \emph{marking} of an and-or graph $\andorgraph$ is a subgraph $\andorgraph'$ of $\andorgraph$ such that:

\begin{itemize}

\item the root of $\andorgraph$ is the root of $\andorgraph'$.

\item if $v$ is a node of $\andorgraph'$ and is an or-node of $\andorgraph$ 
then there exists exactly  one edge $(v, w)$ of $\andorgraph$ that is an edge of $\andorgraph'$.

\item if $v$ is a node of $\andorgraph'$ and is an and-node of $\andorgraph$ then every edge $(v, w)$ of $\andorgraph$
 is an edge of $\andorgraph'$.

\item if $(v, w)$ is an edge of $\andorgraph'$ then $v$ and $w$ are nodes of $\andorgraph'$.

\end{itemize}

\noindent
A marking $\andorgraph'$ of an and-or graph $\andorgraph$ for $\KB{\Tb}{ \Ab}{ \Mb}$ is \emph{consistent} if 
it does not contain any node with label  $\derbot$.
\end{definition}

\begin{example}
\label{ex:exampleGraph}
Figure \ref{fig:firstGraph} shows the and-or graph for the following knowledge base. 

\begin{center}
\begin{tabular}{ccc}
\begin{tabular}{l}
\textbf{Tbox} \\
\fbox{$\top \sqsubseteq \exists S.A  $}
\end{tabular} 
&
\begin{tabular}{l}
\textbf {Abox} \\
\fbox{
$\exists R.A(d)$\ \
$\forall R.\neg B(d)$
}
\end{tabular}
&
\begin{tabular}{l}
\textbf {Mbox} \\
\fbox{
$a \eqm  A$ \ \ 
$b \eqm B$
}
\end{tabular}
\end{tabular}
\end{center}
To reduce the number of nodes, 
we apply lazy unfolding   (denoted by \emph{l. unf.} in the figure)~\cite{Horrocks03}.
Instead of adding $(A \sqcup \neg B) \sqcap (B \sqcup \neg A)$ in the $\equal$-rule, we add $A \equiv B$. Then,
 we do lazy unfolding and replace A by B in expressions of
 the form $A$ or $\neg A$ that appear in $\Xs$ or
 in expressions of the form $A(x)$ or $\neg A(x)$ that
 appear in $\Ab$. 
\end{example}

\begin{figure}
\centering
\includegraphics[width = 0.75\linewidth]{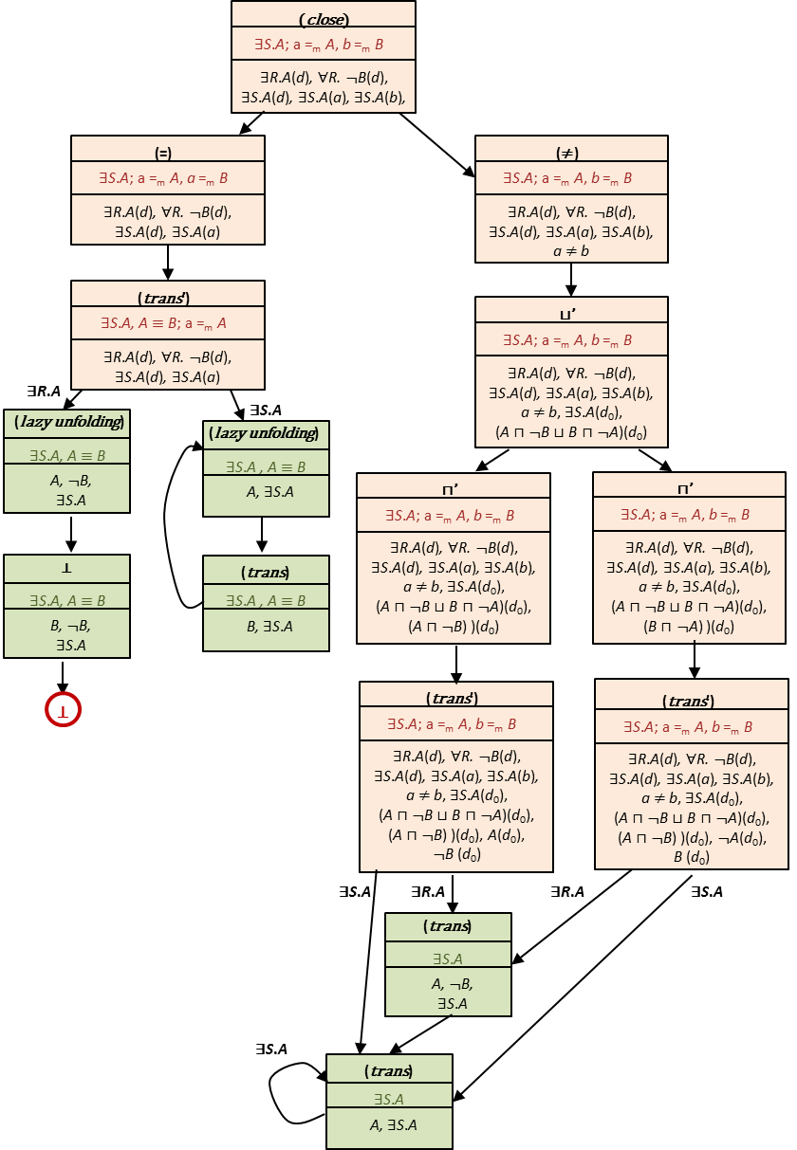}
\caption{And-or graph for the knowledege base: 
$\top \sqsubseteq \exists S.A$, $\exists R.A(d)$, $\forall R.\neg B(d)$, $a \eqm A$, $b \eqm B$ }
\label{fig:firstGraph}
\end{figure}

\section{Correctness of the Tableau Calculus for $\ALCM$}
\label{Section:Correctness}

In this section we prove correctness of the
Tableau Calculus for $\ALCM$ presented in the previous section.

\begin{theorem}[Correctness of the Tableau Calculus] 
\label{theorem:correctnesstableau}
The $\ALCM$ knowledge base  $\Kb = \KB{\Tb}{\Ab}{\Mb}$ is consistent
if and only if the and-or graph for  $\Kb$
 has a consistent marking.
\end{theorem}

\noindent
The if direction is proved in Theorem \ref{theorem:soundess}
(Soundness)
and the only if direction is proved in 
Theorem \ref{theorem:completeness} (Completeness).
We give an idea of the proof.
For the if direction, it is necessary to prove that
the rules preserve satisfiability. A consistent marking
exists because we start from a satisfiable node with label
$\Kb = \KB{\Tb}{\Ab}{\Mb}$. 
For the only if direction, we prove that the converse of base rules
preserve satisfiability, i.e. if the conclusion is satisfiable, so is
the premise.
Then, it is enough to  construct a canonical model
for a base node where no base rules are applicable.
We know at that point that the Mbox has no circularities
which is crucial to enforce that the model is well-founded.
We also know that 
the $\equal$ and $\difference$-rules cannot longer be applied,
which means that there is an isomorphism between
the canonical interpretation of $\ALC$ and $\ALCM$.
\\
 
 We illustrate the idea of  canonical interpretation with the following example. Suppose we have an ontology $(\mathcal{T}, \mathcal{A}, \mathcal{M})$  with four individuals $a, b, c$ and $d$
with axioms $B(a)$, $A(c)$, $ A(d)$ and the meta-modelling axioms   given by
$a  \eqm  A $ and $b  \eqm  B $.
The canonical interpretation $\interp$ of the $\mathcal{ALC}$ ontology 
is then,
\[
\begin{array}{ll}
\Delta^{\interp} & =\{ a, b, c, d \} \\
A^{\interp} &   = \{c, d \} \\
B^{\interp} &  =\{a\}
\end{array}
\]
Intuitively, we see that we need to force the following equations to 
 make the meta-modelling axioms 
$a  \eqm  A $ and 
$b  \eqm  B $ satisfiable:
\begin{center}
\begin{tabular}{l}
$a = \{c, d \}$ \\
$b = \{a\}$
\end{tabular}
\end{center}
These equations do not have circularities. 
Then, the canonical   interpretation
$\interp'$ for the ontology in $\mathcal{ALCM}$ is now defined as follows.
\[\begin{array}{ll}
\Delta^{\interp'} 
& = \{\{c, d \}, \{ \{c,d \} \}, c, d \} \\
A^{\interp'} &   = a^{\interp'} = \{c, d \} \\
B^{\interp'} &  = b^{\interp'} = \{ \{c, d \} \}
\end{array}
\]
In this case, $\interp'$ is a model of $(\mathcal{T}, \mathcal{A}, \mathcal{M})$.
By defining  $S_0 = \{c, d\}$, 
we see that  $\Delta^{\interp'}  \subset S_2$.

\subsection{Soundness}

In order to prove soundness, the following lemma is crucial:

\begin{lemma}[Preservation of Satisfiability in the Tableau Calculus]
\label{lema:preservationrules}
All the rules of Figures \ref{figure:plainrules} and
\ref{figure:primerules} preserve satisfiability, i.e.
if the premise is satisfiable so is the conclusion.
\end{lemma}

\begin{proof}
Note that the statement holds trivially for the bottom rules
since their premises are unsatisfiable. 
  We now give the proof of  the statement for some of the rules that are most interesting. 
  
\begin{description}

	\item[$\closePR$-rule.]  Suppose the premise   
	$\derboxes{\Ts}{\As}{\Mb}$ is satisfiable.
	Then,  there 
	 exists an interpretation 
	 $\interp $  such that 
	 \[\begin{array}{rl}
	  \bullet & \interp \models \bigsqcap \Ts \equiv \top		 \\ 
	  \bullet & \interp \models \As \\
	 \bullet & \interp \models \Mb
	 \end{array}
	 \]
		
	According to the interpretation $\interp$, given two individuals $a,b$ belonging to $\dom (\Mb)$ we have that either 	$a^\interp = b^\interp$ or $a^\interp \not = b^\interp$.\\
	
	We will prove that in both cases 
	$\derboxes{\Ts}{\As [a/b]}{\Mb[a/b]} \ornode   \derboxes{\Ts}{\As \cup \{a\neq b\}}{\Mb}$
	is satisfiable.

	\begin{enumerate}
		\item  Suppose  $a^\interp = b^\interp$.
		  It is easy to prove the following properties:
		\[ \interp \models \As, a^\interp = b^\interp \mbox{ implies } \interp \models \As[a/b] \]
		\[ \interp \models \Mb, a^\interp = b^\interp \mbox{ implies } \interp \models \Mb[a/b] \]
		 		
		 Hence, 
		$ \derboxes{\Ts}{\As[a/b]}{\Mb[a/b]}$ is satisfiable and so is  
		$\derboxes{\Ts}{\As [a/b]}{\Mb[a/b]} \ornode   \derboxes{\Ts}{\As \cup \{a \neq b\}}{\Mb}$.\\

		\item  Suppose  $a^\interp \neq b^\interp$.		
		 This means that  $\interp \models a \neq b$ and then $\interp \models \As \cup \{a \neq b\}$. 		
		Therefore, 
		$\derboxes{\Ts}{\As \cup \{ a \neq b \}}{\Mb}$ is satisfiable
		 and so is 
		$\derboxes{\Ts}{\As[a/b]}{\Mb[a/b]} \ornode   \derboxes{\Ts}{\As \cup \{a \neq b\}}{\Mb}$.\\
	\end{enumerate}
	
	\item[$\equal$-rule.]  
	Suppose the premise   
	$\derboxes{\Ts}{\As}{\Mb \cup \{ a \eqm A, a \eqm B\} }$ is satisfiable.
	Hence,  there 	 exists an interpretation 
	 $\interp $  such that 
	 \[\begin{array}{rl}
	 \bullet & \interp \models \bigsqcap \Ts \equiv \top	 \\
	  \bullet & \interp \models \As \\
	 \bullet & \interp \models \Mb \cup \{ a \eqm A, a \eqm B\}
	 \end{array}.
	 \]
	We need to prove that:
	
	\begin{enumerate}
	\item $ \interp \models \bigsqcap (\Ts \cup \{ (A \sqcup \neg B) \sqcap (B \sqcup \neg A)) \equiv \top$
	\item $\interp \models \As  \cup  \{ ((A \sqcup \neg B) \sqcap (B \sqcup \neg A))(d)  
\suchthat d \in \dom (\As) \cup \dom (\Mb) \cup \{ a \} \}$
	\item $\interp \models \Mb \cup \{a \eqm A\}$\\
	\end{enumerate}

	\begin{enumerate}
	\item Since $\interp \models a \eqm A$, we have that
	\begin{equation} 
	\label{eq:one}
	 a^\interp = A^\interp
	 \end{equation}
	 and since $\interp \models a \eqm B$, we also  have that
	 \begin{equation}
	\label{eq:two}
	 a^\interp = B^\interp
	 \end{equation}
	From (\ref{eq:one}) and (\ref{eq:two}), we get
	\begin{equation}
	 \label{eq:three}
	 A^\interp = B^\interp
	\end{equation}\\
	
	We will prove that $\interp$ validates $(A \sqcup \neg B) \sqcap (B \sqcup \neg A)$, that is
	\[
	((A \sqcup \neg B) \sqcap (B \sqcup \neg A))^\interp = \Delta^\interp.\]
	Applying the definition of interpretation we have that
	\[ ((A \sqcup \neg B) \sqcap (B \sqcup \neg A))^\interp = 
	(A^\interp \cup (\Delta \backslash B^\interp )) \cap (B^\interp \cup (\Delta \backslash A^\interp ) \]
	From (\ref{eq:three}) and replacing in the above expression we obtain
	\[ ((A \sqcup \neg B) \sqcap (B \sqcup \neg A))^\interp = 
	(A^\interp \cup (\Delta \backslash A^\interp )) \cap (A^\interp \cup (\Delta \backslash A^\interp ) = \Delta^ \interp\]
	So, $\interp \models (A \sqcup \neg B) \sqcap (B \sqcup \neg A) \equiv \top$ and then 
	\[
	\interp \models \bigsqcap (\Ts \cup \{ (A \sqcup \neg B) \sqcap (B \sqcup \neg A)) \equiv \top.
	\]
	
	\item Since $((A \sqcup \neg B) \sqcap (B \sqcup \neg A))^\interp = \Delta^\interp $,  we  have that
	\[
	 d^\interp \in ((A \sqcup \neg B) \sqcap (B \sqcup \neg A))^\interp 
	 \]
	  for all 
	$d \in \dom(\As) \cup \dom(\Mb) $.\\
	
	\item $\Mb \cup \{ a \eqm A\} \subset \Mb \cup \{ a \eqm A, a \eqm B\}$ and 
	
	$\interp \models \Mb \cup \{ a \eqm A, a \eqm B\}$, so $\interp \models \Mb \cup \{ a \eqm A \}$
	
	\end{enumerate}
	This means that the conclusion of the $\equal$-rule 
	\[\derboxes{\Ts_1 }{\As_1 }{\Mb \cup \{a \eqm A \}}\] 

where  $\Ts_1  := \Ts \cup \{ (A \sqcup \neg B),  (B \sqcup \neg A) \}$ and 
 $\As_1 =  \As  \cup  \{ ((A \sqcup \neg B) \sqcap (B \sqcup \neg A))(d)  
\suchthat d \in \dom (\As) \cup \dom (\Mb) \cup \{ a \} \}$
is satisfiable. \\

	\item[$\difference$-rule.]
	Suppose the premise   
	$\derboxes{\Ts}{\As \cup \{a \neq b \} }{\Mb \cup \{ a \eqm A, b \eqm B \}}$ is satisfiable.
	Consequently,  there exists an interpretation 
	 $\interp $  such that 
	 \[\begin{array}{rlll}
	 \bullet & \interp \models \bigsqcap \Ts \equiv \top	 \\
	  \bullet & \interp \models \As, a \neq b  &\mbox{thus  } &  \interp \models \As \mbox{ and } \interp \models a \neq b  \\
		\\
	 \bullet & \interp \models \Mb,   a \eqm A, b \eqm B  & \mbox{thus } & \interp \models \Mb, \interp \models  a \eqm A \mbox{ and } \interp \models  b \eqm B\\
																					
	 \end{array}.
	 \]
	
	In order to prove that the conclusion of the rule is satisfiable we only need to prove that
		\begin{enumerate}
		\item $\interp \models (A \sqcap \neg B \sqcup \neg A \sqcap B)(d_0)$ where $d_0$ is a new individual and
		\item $\interp \models C(d_0)$ for all $C \in \Ts$ 
	\end{enumerate}
	as the other conditions follow immediately from the satisfiable property of the premise.\\

	\begin{enumerate}
		\item Since $\interp \models a \neq b$, we have that
					\begin{equation}
					\label{eq:four}
						a^\interp \neq b^\interp
					\end{equation}
					
					Since $\interp \models a \eqm A$ and $\interp \models  b  \eqm B$, we also have that
					\begin{equation}
					\label{eq:five}
						 a^\interp = A^\interp  \ \ \ \ \
						  b^\interp  = B^\interp
						\end{equation}	
						
					It follows from (\ref{eq:four}) and (\ref{eq:five}) that the sets $A^\interp$ and $B^\interp$ are different.
	Hence, there is an element in one of these sets that is not in the other.	
	
	That is to say, there exists an element on the domain of the interpretation 
	$(x \in \Delta^\interp)$ that belongs to the set 
	$(A^\interp \cap (\Delta \backslash B ^\interp)) \cup ((\Delta \backslash A ^\interp) \cap B^\interp)$.
		
	If we use $d_0$, a new individual in $\Kb$, to denote $x$ $({d_0}^\interp = x)$ 
	we have that  $\interp \models (A \sqcap \neg B \sqcup \neg A \sqcap B)(d_0)$.\\

	\item 
Since $\interp \models \bigsqcap \Ts \equiv \top$ we have that $\interp$ validates every concept $C \in \Ts$. This means that $x \in C^\interp$ for all $x \in \Delta^\interp$ 	and for all $C \in \Ts$. In particular: ${d_0}^\interp \in C^\interp$ and hence, $\interp \models C(d_0)$ for all $C \in \Ts$. 
		
	\end{enumerate}
	
	In this way we can conclude that the $\difference$-rule preserves satisfiability.

\end{description}

\end{proof}

\begin{lemma}[Preservation of Satisfiability in the And-or-graph]
\label{lema:preservationnodes}
In an And-or graph $\andorgraph$, for every no-end node with satisfiable label:
\begin{itemize}
	\item if it is an and-node, all its successors have a satisfiable label.
	\item if it is an or-node, at least one of its successors has a satisfiable label.
\end{itemize}
 \end{lemma}

\begin{proof}
Let $v$ be a no-end node with label $E_1$ satisfiable.
The successors of $v$ are obtained by applying one of the rules of  Figures   \ref{figure:plainrules} and
\ref{figure:primerules} (let it be $r$) using $E_1$ as premise of the rule. In addition, the labels of the successors nodes are obtained from the conclusion of the rule.
\begin{description}
	\item[$v$ is an and-node] The conclusion of $r$ is a and-judgement of the form: $\judgement_1 \andnode  \ldots \andnode \judgement_k$, where  $\judgement_i$  are simple judgements for all $1 \leq i \leq k$.
		
	Applying  Lemma   \ref{lema:preservationrules} we know that $\judgement_1 \andnode  \ldots \andnode \judgement_k$ is satisfiable since $E_1$ (premise of $r$) is satisfiable.
	So each  $\judgement_i$ is satisfiable.
	Based on the construction of the graph $\andorgraph$ the labels of the successors of $v$ are each $\judgement_i$.
	Consequently, the labels of all the successors of $v$ are satisfiable.
	
	\item [$v$ is an or-node] The conclusion of $r$ is a or-judgement of the form: $\judgement_1 \ornode \judgement_2$,  where  $\judgement_1$ and $\judgement_2$ are simple judgements.
	Since $E_1$ (premise of $r$) is satisfiable, so is $\judgement_1 \ornode \judgement_2$ (applying 
	 Lemma  \ref{lema:preservationrules}). From the definition of satisfiable we know that $\judgement_1$ is satisfiable or $\judgement_2$ is satisfiable.
	
	In these cases the node $v$ has two successors $w_1$, $w_2$ and the labels of these nodes are $\judgement_1$ and $\judgement_2$ respectively. Hence, at least one of the successors of $v$ has a label which is satisfiable.
\end{description}
\end{proof}

\begin{theorem}[Soundness of the Tableau Calculus of $\ALCM$]
\label{theorem:soundess}
Let $\Kb = \KB{\Tb}{\Ab}{\Mb}$
be a knowledge base of $\ALCM$ in negation normal form.
If $\Kb$ is consistent
then the and-or graph for  $\Kb$
 has a consistent marking.
\end{theorem}

\begin{proof}
Let $\andorgraph$ be the and-or graph for $\Kb$.
We will construct $\andorgraphP$ a consistent marking of $\andorgraph$.
\begin{itemize}
\item First of all, we initialize $\andorgraphP$ with the root of $\andorgraph$.
If  $\Kb = \KB{\Tb}{\Ab}{\Mb}$ is consistent then the label of this node calculated as stated in definition \ref{definition:andorgraph} is satisfiable.

\item Then for each node $v$ in $\andorgraphP$:
		\begin{itemize}
			\item if $v$ is an and-node we add all the successors of $v$ in $\andorgraph$ to $\andorgraphP$, and all the edges between them.
			Applying  Lemma  \ref{lema:preservationnodes}, we deduce that all these new nodes in $\andorgraphP$ have satisfiable labels.
			
			\item if $v$ is an or-node we add to  $\andorgraphP$ the successors of $v$ in $\andorgraph$ having a satisfiable label. We know this node exists from  Lemma  \ref{lema:preservationnodes}.
	
	\end{itemize}
\end{itemize}

It is easy to show that $\andorgraphP$ is a marking of $\andorgraph$ (see Definition \ref{definition:marking}).
The node we use to initialize $\andorgraphP$ has a satisfiable label and we ensure that each node we add has a satisfiable label, so in $\andorgraphP$ there is no node with label $\derbot$. Consequently $\andorgraphP$ is a consistent marking of $\andorgraph$.
Since $\andorgraphP$ is a subgraph of a finite graph (see Remark \ref{remark:graph}),  it is finite.

\end{proof}

\subsection{Completeness}

We define three notions of saturated \graph{\setRoles} \
1) for a Tbox $\Tb$, 2) for an $\ALC$-knowledge base $\KbALC{\Tb}{\Ab}$
and 3) for an $\ALCM$-knowledge base $\KB{\Tb}{\Ab}{\Mb}$.
These three notions are  conditions that ensure satisfiability
of a set $\Xs$ of concepts w.r.t. $\Tb$, 
satisfiability  of a $\ALC$-knowledge base
 $\KbALC{\Tb}{\Ab}$ and satisfiability  of  an $\ALCM$-knowledge base
 $\KB{\Tb}{\Ab}{\Mb}$ respectively
(see Theorems 
\ref{theorem:saturatedgraphvariable}, 
\ref{theorem:saturatedgraphALC} and  \ref{theorem:saturatedgraphALCM}).
They correspond to the so-called ``tableau'' in Description Logic~\cite{Horrocks1999,Motz2015}.
A  saturated \graph{\setRoles} can be seen as an abstraction of a model. 
The weaker notion of consistent model graph given by Nguyen and Szalas  \cite{Nguyen2009} 
includes
clauses 1-5 of Definition \ref{definition:saturatedgraph}
but not the last one. 
Hence, 
it only ensures satisfiability of a set of concepts with respect
to an empty Tbox.
The sufficient conditions for satisfiability (Theorems 
\ref{theorem:saturatedgraphvariable}, 
\ref{theorem:saturatedgraphALC} and  \ref{theorem:saturatedgraphALCM})
structure  the proof of completeness and make it more neat.

\begin{definition}[\graph{\setRoles}]
\label{definition:Rgraph}
We say that $(\Delta, \lab, \Eps)$ is an {\em \graph{\setRoles}} \  if

\begin{itemize}
\item
$\Delta$ is a non-empty set,
\item
$\lab$  maps each element in $\setS$ to a set of concepts, 

\item
 $\Eps : \setRoles \rightarrow 2^{\setS \times \setS}$ 
maps each role in $\setRoles$  to a set of pairs of elements in $\setS$.

\end{itemize}   

\end{definition}

\begin{definition}[Saturated \graph{\setRoles}]
\label{definition:saturatedgraph}
We say that $\T = (\setS, \lab, \Eps)$ is a 
{\em saturated \graph{\setRoles} } \ for a Tbox $\Tb$ if
$\T = (\setS, \lab, \Eps)$ is an \graph{\setRoles}
that  satisfies the following properties
for all $x,y \in \setS$, $R \in \setRoles$ and
 concepts  $C, C_1, C_2$.
\begin{enumerate}
 \item \label{condition:negRestructuraTb}
 If $\neg C \in \lab(x)$, then $ C \not \in \lab(x)$.

 \item \label{condition:capRestructuraTb}
 If $C_1 \sqcap C_2 \in \lab(x)$, then $C_1 \in \lab(x)$ and $C_2 \in \lab(x)$. 
 \item  \label{condition:cupRestructuraTb}
 If $C_1 \sqcup C_2 \in \lab(x)$, then $C_1 \in \lab(x)$ or $C_2 \in \lab(x)$. 
 \item \label{condition:forallRestructuraTb}
 If $\forall R.C \in \lab(x)$ and $(x,y) \in \Eps(R)$, then $C \in \lab(y)$. 
 \item \label{condition:existsRestructuraTb}
  If $\exists R.C \in \lab(x)$, 
  then there is some $y \in \setS$ such that $(x, y) \in \Eps(R)$ and $C \in \lab(y)$. 
	
 \item \label{condition:TbRestructuraTb}
 If $C \in \Tb$ then $C \in \lab(x)$ for all $x \in \setS$. 
\end{enumerate}

\end{definition}

\begin{definition}
\label{definition:interpretationinducedbygraph}
The interpretation $\interp$ induced by an \graph{\setRoles} is defined
as follows.
\[
\begin{array}{lll}
\Delta^\interp &= \setS \\
A^\interp &  := \{ x \in \Delta \suchthat A \in \lab(x) \}
\\
R^\interp & := \Eps(R)
\end{array}
\] 

\end{definition}

As usual, it is enough for an interpretation as the one given above  to
define it only  for atomic concepts $A$.
Next lemma gives  a characterization of that interpretation for complex concepts.

\begin{lemma}
\label{lemma:conceptinterp}
Let 
$\T = (\setS, \lab, \Eps)$ be a
saturated \graph{\setRoles} \ for
a Tbox $\Tb$ and $\interp$ be the interpretation induced by  $\T$. For every $x \in \setS$, if
$C \in \lab(x)$ then  $x \in C^\interp$.
\end{lemma}

\begin{proof}
This is proved by induction on $C$. We prove some cases. 

\begin{itemize}

\item Suppose $C = A$. This is the base case.
By Definition \ref{definition:interpretationinducedbygraph},
we have that $x \in A^{\interp}$.

\item  Suppose $C = \neg D$. Since $C$ is in
negation normal form, we have that $D $ is an atomic
concept, say $A$.
Since $\T$ is a saturated \graph{\setRoles}, by Definition \ref{definition:saturatedgraph}, $A \notin \lab(x)$.
By Definition \ref{definition:interpretationinducedbygraph},
we have that $x \not \in A^{\interp}$.
Hence, $x \in (\neg A)^{\interp}$.

\item Suppose $C = C_1 \sqcap C_2$. By Definition \ref{definition:saturatedgraph}, $C_1 \in \lab(x)$ and $C_2 \in \lab(x)$. By induction hypothesis we have that $x \in C_1^\interp$ and $x \in C_2^\interp$. Then $x \in C_1^\interp \cap C_2^\interp$, so $x \in (C_1 \sqcap C_2)^\interp$, which means that $x \in C^\interp$.

\item Suppose $C = \exists R.D$. By Definition \ref{definition:saturatedgraph}, 
there is some $y \in \setS$ such that $(x, y) \in \Eps(R)$ and $D \in \lab(y)$. 
By Definition \ref{definition:interpretationinducedbygraph}, we have that $(x, y) \in R^\interp$. 
By induction hypothesis, 
 $y \in D^\interp$. Then $x \in (\exists R.D)^\interp$, and so $x \in C^\interp$.
\end{itemize}

\end{proof}

\begin{theorem}
\label{theorem:saturatedgraphvariable}
Let  $\T = (\setS, \lab, \Eps)$ be a saturated \graph{$\setRoles$} \  
for a Tbox $\Tb$.
If $\Xs \subseteq \lab(x_0)$ for some $x_0 \in \setS$
then $\Xs$ is satisfiable w.r.t. $\Tb$.
\end{theorem}

\begin{proof}
Let $\interp$ be an interpretation induced by $\T$.
It follows from Lemma \ref{lemma:conceptinterp} that 
$x_0 \in C^\interp $ for all $C \in \Xs$.
Hence $(\bigsqcap \Xs)^\interp \not = \emptyset$.
We also have that $\interp$ is a model of $\Tb$
from the last clause of Definition \ref{definition:saturatedgraph}.
\end{proof}

\begin{definition}[Saturated \graph{\setRoles} for $\KbALC{\Tb}{\Ab}$]
\label{definition:saturatedgraphALC}
We say that $\T = (\setS, \lab, \Eps)$ is a 
{\em saturated \graph{\setRoles} } \ for $\KbALC{\Tb}{\Ab}$ if
$\T = (\setS, \lab, \Eps)$ is a saturated  \graph{\setRoles} \
for $\Tb$ that satisfies the following properties.
\begin{enumerate}

\item \label{condition:dominioRestructuraALC} $\Delta$ contains all the individuals of $\Ab$.

 \item  \label{condition:assertIndivRestructuraALC}If $C(a) \in \Ab$, then $C \in \lab(a)$. 

 \item \label{condition:assertRolesRestructuraALC} If $R(a, b) \in \Ab$, then  $(a, b) \in \Eps(R)$. 

 \item \label{condition:diffRestructuraALC} If $a \not = b  \in \Ab$ then $a$ and $b$ are syntactically different.

 \end{enumerate}
 \end{definition}
 
 \begin{theorem}
\label{theorem:saturatedgraphALC}
 Let $\Ab$ an Abox without equality axioms and  
 $\T = (\setS, \lab, \Eps)$ be a 
{\em saturated \graph{\setRoles} } \ for $\KbALC{\Tb}{\Ab}$.
Then, the $\ALC$-knowledge base $\KbALC{\Tb}{\Ab}$ is satisfiable (consistent).

 \end{theorem}
 
 \begin{proof}
Let $\interp$ be the interpretation induced by  $\T$.
For each individual $a$ of $\Ab$, we define
$a^\interp = a$. 
 We will show that $\interp$ is a model for $\KbALC{\Tb}{\Ab}$. \\
$\interp$ is a model of $\Tb$ by the last clause of Definition \ref{definition:saturatedgraph} and Lemma \ref{lemma:conceptinterp}. To show that $\interp$ is a model of $\Ab$, we will prove that $a^\interp \in C^\interp$ for all  $C(a) \in  \Ab$, $(a^\interp, b^\interp) \in R^\interp$ for all $R(a, b)  \in \Ab$ and  $a^\interp \not= b^\interp$  for all $a \not= b  \in \Ab$.

\begin{itemize}

\item Suppose $C(a) \in  \Ab$. By Definition \ref{definition:saturatedgraphALC}, $C \in \lab(a)$ and by Lemma \ref{lemma:conceptinterp}, $a^\interp = a \in C^\interp$.

\item Suppose $R(a, b) \in  \Ab$. By Definition \ref{definition:saturatedgraphALC}, $(a, b) \in \Eps(R)$ and by Definition \ref{definition:interpretationinducedbygraph}, $(a^\interp, b^\interp)  = (a, b) \in R^\interp$.
  
\item Suppose $a \not= b  \in \Ab$. By Definition \ref{definition:saturatedgraphALC}, $a$ and $b$ are  syntactically different, so $a^\interp \not= b^\interp$.

\end{itemize}

 \end{proof}

 \begin{definition}[Circularity of an \graph{\setRoles} w.r.t an Mbox]
\label{definition:circularMbox}
We say that $\T = (\setS, \lab, \Eps)$ has a circularity w.r.t. $\Mb$ 
if   there is a sequence of meta-modelling axioms 
$ a_1  \eqm A_1 $, $ a_2  \eqm A_2 $, $ \ldots$, $ a_n \eqm  A_n$ all in $\Mb$ such that 
$A_1 \in \lab(a_2)$, $A_2 \in \lab(a_3)$, \ldots, $A_n \in \lab (a_1)$. 
\end{definition}

\begin{definition}[Saturated \graph{\setRoles} \ for $\KB{\Tb}{\Ab}{\Mb}$]
\label{definition:saturatedgraphALCM}
We say that $\T = (\setS, \lab, \Eps)$ is a 
{\em saturated \graph{\setRoles} } \ for $\KB{\Tb}{\Ab}{\Mb}$ if
$\T = (\setS, \lab, \Eps)$ is a saturated  \graph{\setRoles} \
for $\KbALC{\Tb}{\Ab}$ that satisfies the following properties.
 \begin{enumerate} 
 
 \item  \label{condition:dominioRestructureALCM} $\Delta$ contains all the individuals of $\Mb$.

 \item \label{condition:circularitiesRestructureALCM} $\T$ has no circularities w.r.t. $\Mb$.
 
  \item \label{condition:eqConceptsRestructureALCM} if $a \eqm A\in \Mb$ and 
$a \eqm B\in \Mb$, then   $A = B$, i.e. $A$ and $B$ are syntactically equal. 

  \item \label{condition:difIndivRestructureALCM} if $a$ and $b$ are syntactically different, 
 $a \eqm A\in \mathcal{M}$ and $b \eqm B\in \mathcal{M}$, then there is some $t \in \setS$ such that $A \sqcap \neg B \sqcup B \sqcap \neg A \in \lab(t)$. 
 \end{enumerate}
\end{definition}

Note that the set $\setS$ of  a saturated \graph{\setRoles} for $\KB{\Tb}{\Ab}{\Mb}$ 
contains the individuals of $\Mb$
as well as the ones of $\Ab$.
In Theorem~\ref{theorem:constructionsaturatedstructure}
we  construct a \graph{\setRoles} for
a knowledge base $\Kb = \KB{\Tb}{\Ab}{\Mb}$
where no other rule is applicable to $\Kb$ except for the 
 $\transPR$-rule. In particular,  $\Ab$ does not contain equalities
and $\Mb$ does not contain two axioms with the same individual $a$, i.e.
if $a \eqm A$ and $a \eqm B$ are in
$\Mb$ then $A$ and $B$ are syntactically equal. 

 \begin{definition}[From Basic Objects to Sets]
 \label{definition:unfold}
Let  $\T = (\setS, \lab, \Eps)$ be a  saturated \graph{$\setRoles$}
for $\KB{\Tb}{\Ab}{\Mb}$. 
For  $x \in \setS$  we define $\unfold(x)$ as follows. 
\[
\begin{array}{lll}
\unfold(x)  = & x  &  \mbox{ if } x \not \in \dom(\Mb)  \\
\unfold(x) = & \{ \unfold(y) \mid A \in \lab(y)  \} 
  & \mbox{ otherwise, i.e. $x\in \dom(\Mb)$ and $x \eqm A \in \Mb$ } 
\end{array}
\] 
\end{definition}
Since  $\setS$ contains the individuals of $\Mb$,
 an element of $\setS$ either belongs to $\dom(\Mb)$ or not.
In case, $x \in \dom(\Mb)$ then we have that $x \eqm A \in \Mb$.

\begin{example}
 We consider the example of Figure \ref{figure:boxesforfirstview} and
 add a new individual $hydrographic$ and the meta-modelling axiom

 \[
 hydrographic \eqm HydrographicObject
 \]
 Here we have for example that $river$ is an individual with
 meta-modelling. As such, its interpretation should be a set
 and not a basic object. The set associated to $river$  is given
 by the function $\unfold$ and it is as follows.
 \[
 \begin{array}{lll}
 \unfold{(river)} & = & \{queguay, santaLucia\}
 \end{array}
 \]
 The individual $hydrographic$ has also meta-modelling. But
 its inhabitants also have meta-modelling. 
 The set associated to $hydrographic$ is a set of sets
 given as follows. 
 \[
 \begin{array}{lll}
 \unfold{(hydrographic)} &
 = & \{\{queguay, santaLucia\}, \\
 & & \{deRocha, delSauce\} \} 
 \end{array}
 \]
  On the other hand, $queguay$ does not have meta-modelling
  and we define $\unfold$ as follows.
 \[ \unfold{(queguay)} = queguay.\]
  \end{example}

\begin{theorem}[Correctness of the recursive definition]
\label{theorem:recursion}
Let $\T = (\setS, \lab, \Eps)$ be a 
 saturated \graph{$\setRoles$}
for $\KB{\Tb}{\Ab}{\Mb}$.  
The function $\unfold$ is a correct recursive definition. 
\end{theorem}

\begin{proof} It follows from 
the third clause of Definition \ref{definition:saturatedgraphALCM} that
$\unfold$ is indeed a function. 
In order to prove that is a correct recursive definition, we
apply the recursion principle on well-founded relations (see Definition \ref{definition:recursionprinciple}).
\\  
For this,  we define the relation 
$\prec$  on  $\setS$  as 
 $y \prec a$  iff
$A \in \forest(y)$ and $a \eqm A \in \Mb$.
Since 
$\forest$ has no circularities
w.r.t.\ $\Mb$, it is easy to prove that
$\prec$ is well-founded. 
Suppose towards a contradiction that 
$\prec$  is not well-founded. 
It follows from Lemma  \ref{lemma:decreasingsequences},
that  there exists an infinite $\prec$-decreasing sequence, 
i.e. there is $\langle x_{i}\rangle_{i \in \mathbb{N}}$ such 
that $x_{i+1} \prec  x_i$ and $x_i \in \Delta$ for all $i \in \mathbb{N}$.
 \[ \ldots \prec x_{i+1} \prec x_i \prec \ldots \prec x_1
 \]
 By definition of $\prec$, 
 we have that
 $x_{i} \in \dom(\Mb)$ for all $i \in \mathbb{N}$.
  Since the Mbox is finite,
 there exists an element  in the above sequence 
 that should occurs at least twice, i.e. 
 $x_{i+1} = x_{i+n+1}= a_1$ for some $i, n \in \mathbb{N}$. Hence, 
 we have a cycle
 \[  \equalto{x_{i+n+1}}{a_1}  \prec \equalto{x_{i+n}}{a_n}  \prec  \equalto{x_{n+i-1}}{a_{n-1}}
  \prec \ldots \prec \equalto{x_{i+2}}{a_2} \prec \equalto{x_{i+1}}{a_1} 
 \]
It is easy to see that 
 this contradicts the fact
 that $\T$ has no circularities w.r.t. $\Mb$.
 \\
 Since $\prec$ is well-founded, 
we can  apply the recursion principle  
in Definition \ref{definition:unfold}.
 Note that in the recursive step of that definition, we have that
$y \prec x$. 
\end{proof}

\begin{lemma}
\label{lemma:setisinSn}
Let $\T = (\setS, \lab, \Eps)$ be a 
saturated \graph{$\setRoles$}
for $\KB{\Tb}{\Ab}{\Mb}$.  

If  $S_0 = \{ x \in \setS \mid x \notin \dom(\Mb) \}$ 
 then
\[\unfold(x) \in S_{\sharp \Mb}\]
\end{lemma}

\begin{proof}
 Let $\maxprec(x)$ 
be  the {\em maximal length} 
 of all descending $\prec$-sequences starting from $x \in \setS$.
 This number is finite because $\prec$ is well-founded by
 Theorem \ref{theorem:recursion}. It is not difficult  to show that
 \begin{center}
$\maxprec(x) \leq \sharp\Mb$ for all  $x$ in $\setS$.
\end{center}
It is also easy to prove that
\[
\unfold (x) \in S_{\maxprec(x)}
\]
by induction on $\maxprec(x)$.
Then $\unfold(x) \in S_{\maxprec(x)}
\subseteq S_{\sharp\Mb}$ since
$S_n$ is a monotonic function on $n \in \mathbb{N}$ (see Definition \ref{definition:domainSet}).
\end{proof}

\begin{lemma}
\label{lemma:Igualsintsema}
 Let $\Ab$ be an Abox without equality axioms,
 $\T = (\setS, \lab, \Eps)$ be a 
{\em saturated \graph{\setRoles} } \ for $\KB{\Tb}{\Ab}{\Mb}$,
$\{a \eqm A,b \eqm B \} \subseteq \Mb$ and $\interp$ the interpretation induced by $\T$. 
Then,  $a$ and $b$ are syntactically equal if and only if   $A^\interp = B^\interp$.
\end{lemma}

\begin{proof}

\begin{itemize}

\item if 
 $a$ and $b$ are syntactically equal 
then $a \eqm A$ and $a \eqm B$ are both in $\Mb$. 
Applying the Definition 
\ref{definition:saturatedgraphALCM} we have that $A=B$ and so are their interpretations.\\

\item if $a$ and $b$ are syntactically different, applying 
 Definition \ref{definition:saturatedgraphALCM}, 
 it follows that there is some $t \in \setS$ such that $A \sqcap \neg B \sqcup B \sqcap \neg A \in \lab(t)$.
	Applying  Definition \ref{definition:saturatedgraph} we have two possibles cases:
	 \begin{enumerate}
		\item $A\in \lab(t)$ and $B \not \in \lab(t)$, so $t \in A^\interp$ and $t \not\in B^\interp$, then 
			$A^\interp \not\subseteq B^\interp$.
			
		\item $A \not \in \lab(t)$ and $B\in \lab(t)$, so $t \in B^\interp$ and $t \not\in A^\interp$, then 
			$B^\interp \not\subseteq A^\interp$.\\
	\end{enumerate}
	
	In both cases, $A^\interp \neq B^\interp$ 

\end{itemize} 

\end{proof}

\begin{theorem} \label{theorem:saturatedgraphALCM}
 Let $\Ab$ an Abox without equality axioms and  
 $\T = (\setS, \lab, \Eps)$ be a 
{\em saturated \graph{\setRoles} } \ for $\KB{\Tb}{\Ab}{\Mb}$.
Then,  $\KB{\Tb}{\Ab}{\Mb}$ is satisfiable (consistent).
\end{theorem}

\begin{proof}
The interpretation $\interp$ induced by $\T$ is 
a model of $\KbALC{\Tb}{\Ab}$ by Theorem \ref{theorem:saturatedgraphALC}.
We define the interpretation $\interp'$ as follows.
\[
\begin{array}{ll}
\Delta^{\interp'} & = \{ \unfold(x) \suchthat x \in \Delta^\interp \} \\
A^{\interp'} & = \{ \unfold (x) \suchthat x \in A^{\interp} \} \\
R^{\interp'} & = \{ (\unfold(x), \unfold(y) ) \suchthat (x, y) \in R^\interp \}\\
a^{\interp'} & = \unfold (a^{\interp})
\end{array}
\]
We prove the three clauses of Definition \ref{definition:modelALCM}.

\begin{enumerate}

\item 
It follows from Lemma \ref{lemma:setisinSn} that
\[\Delta^{\interp'} \subseteq  S_{\sharp \Mb}\]

\item  
By Lemma \ref{lemma:isomorphism},   
$\interp'$ is a model of $\KbALC{\Tb}{\Ab}$ because $\interp$ and $\interp'$ are isomorphic interpretations.
\begin{center}
{\bf Claim.} $\interp$ and $\interp'$ are isomorphic interpretations.
\end{center}
In order to  prove the claim, we prove
 that $\unfold$ is a bijection between the domains of $\interp$ and $\interp'$ as follows.

\begin{itemize}
\item It follows from Theorem \ref{theorem:recursion} that $\unfold$ is a function. 

\item It is surjective since the elements of the domain of $\interp'$ are  defined as the result of applying $\unfold$ to the elements of the domain of $\interp$.

\item \ We prove that $\unfold $ is injective by induction on $\rel$
by appling the Induction Principle given in Definition \ref{definition:inductionprinciple}. 
We  prove that if $x_1$ and $x_2$ are two different elements in $\Delta$
then $\unfold(x_1) \not = \unfold(x_2)$, or equivalently that
if $\unfold(x_1)  = \unfold(x_2)$ then $x_1 = x_2$.
Suppose now that  $x_1$ and  $x_2$ are two different elements of $\Delta$.
\begin{itemize}
 \item {\bf Base Case}. Suppose $x_1,x_2$ are individuals without meta-modelling. 
 It   follows from the definition of $\unfold$ that  $\unfold(x_1) \neq \unfold(x_2)$.
 \item  {\bf Base Case}. Suppose only one of them has meta-modelling, say $x_1$.  
 We know that $\unfold(x_1) \not\in \Delta$ since it is a set but $\unfold(x_2) \in \Delta$ (is $x_2$) so $\unfold(x_1) \neq \unfold(x_2)$.
	\item  {\bf Inductive Case}.
	
	Suppose that $x_1$ and $x_2$ are individuals with meta-modelling:
	$x_1 \eqm A_1$ and $x_2 \eqm A_2$.
	It follows from the definition of $\unfold$  that:	
	\[\unfold(x_1) = \{ \unfold(q) \suchthat A_1 \in \forest(q) \}\]
	\[\unfold(x_2) = \{ \unfold(q') \suchthat A_2 \in \forest(q') \}\]
	 Suppose towards a contradiction that  $\unfold(x_1) = \unfold(x_2)$.   Hence, 
	$\unfold(x_1) \subseteq \unfold(x_2)$ and $\unfold(x_2) \subseteq \unfold(x_1)$.
	Since $\unfold(x_1) \subseteq \unfold(x_2)$,
	for all $q$  such that $A_1 \in \forest(q)$ there exists $q'$ such that
	$A_2 \in \forest(q')$ and $\unfold(q)=\unfold(q')$.\\
	It follows from the definition of $\rel$ that $q \rel x_1$ and $q'\rel x_2$, 
	
	By induction hypothesis,  if $q \not = q'$ then
	$\unfold(q) \not =\unfold(q')$.
	Since $\unfold(q)=\unfold(q')$, we actually have that $q=q'$. 
	\\
	Thus, $\{ q \suchthat A_1 \in \forest(q) \} \subseteq \{ q' \suchthat A_2 \in \forest(q') \}$, which means that:
	
	\begin{equation}
   \label{eq:InjectiveA1inA2}	
		{A_1}^\interp \subseteq {A_2}^\interp
	\end{equation}
	
	Analogously,  from $\unfold(x_2) \subseteq \unfold(x_1)$ we can prove that  
	\begin{equation}
   \label{eq:InjectiveA2inA1}	
		{A_2}^\interp \subseteq {A_1}^\interp
	\end{equation}
	
	It follows from (\ref{eq:InjectiveA1inA2}) and (\ref{eq:InjectiveA2inA1}) that 
	${A_2}^\interp = {A_1}^\interp$. Then, applying Lemma \ref{lemma:Igualsintsema} we have that $x_1 = x_2$
	which contradicts the fact that $x_1$ and $x_2$ are different.
	Hence, $\unfold(x_1) \not = \unfold(x_2)$.

\end{itemize}
In this way we can conclude that $\unfold$ is injective, so it is a bijection and thus $\interp$ and $\interp'$ are isomorphic interpretations.\\
\end{itemize}

\item We prove that ${\interp'}$ is a model of $\Mb$. Suppose $a \eqm A \in \Mb$. 
From the definition of $\interp'$, we know that:

\[a^{\interp'} = \unfold (a^{\interp}).\]

Applying the definition of $\unfold$, we have that:

\[a^{\interp'} = \{ \unfold(y) \suchthat A \in \lab(y) \}\]

Finally, applying the Definition  \ref{definition:interpretationinducedbygraph} we have

\[a^{\interp'} = \{ \unfold(y) \suchthat y \in A^{\interp}  \}\]

which is the definition of $A^{\interp'}$

Thus $\interp' \models a \eqm A $ for all $a \eqm A \in \Mb$, then ${\interp'}$ is a model of $\Mb$.

\end{enumerate}

\end{proof}

\begin{lemma}
\label{lemma:satisfiableSust}
If $\derboxes{\Ts}{\As[a/b]}{\Mb[a/b]}$ is satisfiable, $\{a, b \} \subseteq \dom (\Mb) \cup \dom(\As)$, then $\derboxes{\Ts}{\As}{\Mb}$ is satisfiable.
\end{lemma}

\begin{proof} 
 Suppose  $\derboxes{\Ts}{\As[a/b]}{\Mb[a/b]}$ is satisfiable.
Then, 
 there exists $\interp$ such that $\interp \models \As[a/b]$, 
 $\interp \models \Mb[a/b]$ and $\interp \models \bigsqcap \Ts \equiv \top$.\\

We define  a new interpretation $\interp_1$ that extends $\interp$ by adding 
$b^{\interp_1} = a^{\interp}$. 

\begin{itemize}

\item  Since $\interp \models \bigsqcap \Ts \equiv \top$, we obviously have that

\begin{equation}
		\label{eq:eqTboxSust}
		\interp_1 \models \bigsqcap \Ts \equiv \top.
		\end{equation}\\

\item  The Abox $\As[a/b]$ is obtained from $\As$ by replacing all occurrences of $b$ by $a$. 
Let $\statement \in \As$. We will prove that $\interp_1 \models \statement$.

There are several possibilities:

	\begin{enumerate}
	
	\item Suppose that $\statement \in \As$ does not contain $b$. Then, 
	 it is easy to see that $\interp_1 \models \statement$ because 
$\statement$ and $\statement [a/b]$ are exactly the same statement
and $\interp \models \statement$.

	\item Suppose that $\statement$ is  $C(b)$. Then, $\statement [a/b]$
	is  $C(a)$.  Hence, 
	$b^{\interp_1}  = a^{\interp} 
	               \in C^{\interp} = C^{\interp_{1}} $ since  $\interp \models C(a)$.    
	So $\interp_1 \models C(b)$.
		
	\item   Suppose that $\statement$ is $a = b$ or $b = a$.
	Since $b^{\interp_1} = a^{\interp_1}$, $\interp_1 \models a = b$ and $\interp_1 \models b = a$

	\item Suppose that $\statement$ is $b = c$ and $c$ is not $a$.
	Then, $\statement [a/b]$ is $a = c$ and $\interp \models a = c$.
	Since $b^{\interp_1} = a^{\interp_1} = c^{\interp_1}$, $\interp_1 \models b = c$.
	
	\item Suppose that $\statement$ is $c = b$ and $c$ is not $a$.
	This case is similar to the previous one.

	\item Suppose that $\statement$ is $b \neq c$ and $c$ is not $a$.
	Then, $\statement [a/b]$ is 
	$a \neq c$ and  $\interp \models a \neq c$.
	 Since $b^{\interp_1} = a^{\interp_1}\neq c^{\interp_1}$, 
	  $\interp_1 \models b \neq c$.

	\item Suppose that $\statement$ is $c \neq b$ and $c$ is not $a$. This case is similar to the previous one.

	\item Suppose that $\statement$ is  $a \neq b$ or $b \neq a$.
	But then $\statement [a/b]$ is $a \neq a$. 
	This case is not possible because $a \neq a$ is not satisfiable.

	\end{enumerate}
	Hence, we have just proved that 
		\begin{equation}
		\label{eq:eqAboxSust}
		\interp_1 \models \As.
		\end{equation}

\item  Let $\statement \in \Mb$. We will prove that $\interp_1 \models \statement$.
There are two  possibilities:
\begin{enumerate}
\item Suppose $\statement$ does not contain $b$. Then,
it is easy to see that $\interp_1 \models \statement$.

\item Suppose $\statement$ is  $b \eqm A \in \Mb$.
Then $\statement [a/b]$  is $a \eqm A$  and $\interp \models a \eqm A$.
Thus,
$b^{\interp_1} = a^\interp = A^{\interp} = A^{\interp_1}$ and
$\interp_1 \models b \eqm A$.

\end{enumerate}

We  conclude that

\begin{equation}
		\label{eq:eqMboxSust}
		\interp_1 \models \Mb.
		\end{equation}
		
\end{itemize}

From \eqref{eq:eqTboxSust}, \eqref{eq:eqAboxSust} and \eqref{eq:eqMboxSust} we can affirm that  $\derboxes{\Ts}{\As}{\Mb}$  is satisfiable.
\end{proof}

\begin{lemma}
\label{lemma:conversepreservation}
Suppose there is an edge from  a base node $v$ to  a base node 
$w$ in an and-or graph $\andorgraph$.
If the label of $w$ is satisfiable, so is the label of $v$. 
\end{lemma}

\begin{proof}

As $w$ is a successor of $v$ in and-or graph $\andorgraph$ and $v$ is a base node, the label of $w$ is obtained from the application of some rule $r$ of the Figure  \ref{figure:primerules} to the label of $v$.
Given that $w$ is a base node, $r$ must be one of : $\andPR$,$\orPR$,$\forallPR$,$\closePR$,$\equal$ or$\difference$. We will analyse each one of these cases.

\begin{description}
\item [$\{\andPR,\orPR,\forallPR \}$-rules] - The labels of $v$ and $w$ are of the form: $\derboxes{\Ts}{\As_v}{\Mb}$ and 

$\derboxes{\Ts}{\As_w}{\Mb}$ respectively.
If $\derboxes{\Ts}{\As_w}{\Mb}$ is satisfiable, so is $\derboxes{\Ts}{\As_v}{\Mb}$ since they have the same TBox, MBox and $\As_v$ is strictly contained in $\As_w$.

\item[$\closePR$-rule] - If $\derboxes{\Ts}{\As_v}{\Mb}$ is the label for $v$,   
 there are two possibilities for the label of $w$:
	\begin{enumerate}
	 \item Suppose the label is $\derboxes{\Ts}{\As_v \cup \{ a \neq b \}}{\Mb}$. As in the case above if $\derboxes{\Ts}{\As_v \cup \{ a \neq b \}}{\Mb}$ is satisfiable, so is $\derboxes{\Ts}{\As_v}{\Mb}$ .
	
		\item Suppose the label is $\derboxes{\Ts}{\As_v[a/b]}{\Mb[a/b]}$.
		 This case follows from Lemma \ref{lemma:satisfiableSust}.
		
	\end{enumerate}

	\item[$\difference$-rule] - If the label of $v$ is $\derboxes{\Ts}{\As_v \cup \{ a \neq b\}}{\Mb \cup \{ a \eqm A, b \eqm B \}}$ then the label of $w$ is of the form: $\derboxes{\Ts}{\As_w}{\Mb \cup \{ a \eqm A, b \eqm B \}}$.
	
	So if $\derboxes{\Ts}{\As_w}{\Mb \cup \{ a \eqm A, b \eqm B \}}$ is satisfiable, so is the label of $v$ since they have the same TBox, MBox and $\As_v \cup \{ a \neq b\}$ is strictly contained in $\As_w$.
	
	\item [$\equal$-rule] - The label of $v$ is of the form: $\derboxes{\Ts_v}{\As_v}{\Mb \cup \{ a \eqm A, a \eqm B \}}$, 
	the label of $w$ is $\derboxes{\Ts_w}{\As_w}{\Mb \cup \{ a \eqm A \}}$ where $\Ts_w = \Ts_v \cup \{ (A \sqcup \neg B) \sqcap (B \sqcup \neg A) \}$
	and $\As_w = \As_v \cup  \{ ((A \sqcup \neg B) \sqcap (B \sqcup \neg A))(d) \mid d \in \dom(\As_v) \cup \dom(\Mb) \cup \{a\} \} $. \
	\\
	 Since  the label of $w$ is satisfiable,
	there  exists $\interp$ such that:
	
	\begin{enumerate}
	 \item $\interp$ validates all the concepts in $\Ts_w$, so $\interp$ validates all the concepts in $\Ts_v$ and $(A \sqcup \neg B) \sqcap (B \sqcup \neg A)$.
	 \item $\interp \models \As_w$, so $\interp \models \As_v \cup \{ ((A \sqcup \neg B) \sqcap (B \sqcup \neg A)) (d) \mid d \in \dom(\As_v) \cup \dom(\Mb) \cup \{a\}   \}$, then 
	 $\interp \models \As_v $.
	 \item $\interp \models \Mb \cup \{ a \eqm A \}$, then $\interp \models \Mb$ and $\interp \models a \eqm A$ \\
	\end{enumerate}

	Since $\interp$ validates $(A \sqcup \neg B) \sqcap (B \sqcup \neg A)$ we have that: $((A \sqcup \neg B) \sqcap (B \sqcup \neg A))^\interp = \Delta^\interp$.
	
	Applying the definition of interpretation we have that 
	
	$(A^\interp \cup (\Delta \backslash B^\interp)) \cap (B^\interp \cup (\Delta \backslash A^\interp)) = \Delta^\interp$ , so 	
	$(A^\interp \cup (\Delta \backslash B^\interp)) =  \Delta^\interp$ and 
	
	$(B^\interp \cup (\Delta \backslash A^\interp)) = \Delta^\interp$.
	
	Consequently, $B^\interp \subseteq A^\interp$ and $A^\interp \subseteq B^\interp$, which it means that
	
	  \begin{equation}
	   \label{eq:equalityAB}
	    A^\interp = B^\interp
	  \end{equation}
	  
	  Since $\interp \models a \eqm A$ we have that:
	  \begin{equation}
	   \label{eq:equalityaA}
	    a^\interp = A^\interp
	  \end{equation}

	   It follows from
	   (\ref{eq:equalityAB}) and (\ref{eq:equalityaA}) 
	   that 
	    $a^\interp = B^\interp$, so $\interp \models a \eqm B$.\\
	   
	   Hence, we have proved all the necessary conditions 
	    for the satisfiability of 
	    
	   $\derboxes{\Ts_v}{\As_v}{\Mb \cup \{ a \eqm A, a \eqm B \}}$.
	
\end{description}
\end{proof}

\begin{table}

\setlength{\parindent}{0pt}\line(1,0){344}
\newline
{\bfseries{Algorithm for building a saturated \graph{\setRoles}}} 
\newline
\setlength{\parindent}{0pt}\line(1,0){344}
\newline
{\bfseries{Input}}:  $\andorgraphP$ a consistent marking of and-or graph w.r.t.  $\Kb = \derboxes{\Tb}{\Ab}{\Mb}$ such that
 $\Kb$ satisfies the hypotheses of Theorem \ref{theorem:constructionsaturatedstructure}. 
\newline
{\bfseries{Output}}: A  \graph{\setRoles} $\T = (\setS, \lab, \Eps)$  saturated  for $\KB{\Tb}{\Ab}{\Mb}$.

\begin{enumerate}
\item  Let $v_0$ be the root of $\andorgraphP$ 



$\Delta := \dom(\Ab) \cup \dom(\Mb)$

For each $a \in \Delta$ do
			\begin{itemize}
			 \item $\lab(a) := \{ C \suchthat C(a) \in \Ab \}$
			 \item mark $a$ as unresolved
			 \item  $f(a) = v_0$ 
      \end{itemize}
			
end-for

For each role name $R$ do
	\begin{itemize}
	 \item $\Eps(R) := \{ (a,b) \suchthat R(a,b) \in \Ab \}$
	\end{itemize}
end-for

\item While $\Delta$ 
contains unresolved  elements, select one of them: $x$ and do
 \begin{itemize}
		\item For each $\exists R.C \in \lab(x)$ do
			\begin{itemize}

					 \item $u := f(x)$
			
					 \item Let $w_0$ the node of $\andorgraphP$ such the edge $(u,w_0)$ is labeled by $\exists R.C$

			 \item Let $w_0$ \ldots $w_h$ be the saturation path of $w_0$ w.r.t. $\andorgraphP$.
			
			 \item 
			  Let $\Ys = \bigcup_{i=0}^{h} \Xs_{i}$ where 
			              the labels of $w_0$ \ldots $w_h$ are
			              $\derconcepts{\Ts}{\Xs_0}$, ..., $\derconcepts{\Ts}{\Xs_h}$.

			\item if there  is no  $y \in \Delta$ such that $\lab(y) = \Ys$ then
				\begin{itemize}
				  \item 
				  create a new variable $y$ and 
				  add $y$ to $\Delta$
					\item $\lab(y) := \Ys$
					\item mark $y$ as unresolved
					\item $f(y) := w_h$

				 \end{itemize}
				
			end-if
			\item Add $(x,y)$ to $\Eps(R)$
			\end{itemize}

	end-for
	
	\item Mark $x$ as resolved

\end{itemize}
end-while
\end{enumerate}

\setlength{\parindent}{0pt}\line(1,0){344}

\caption{Algorithm for building a saturated \graph{\setRoles}}

\label{table:algorithmbuildRstructured}

\end{table}

\begin{definition}[Saturation path] 
\label{definition:saturationpath}
Let $\andorgraph$ be an and-or graph with a consistent marking $\andorgraph'$ and let $v$ be a node of $\andorgraph'$.
A \emph{saturation path} of $v$ w.r.t.  $\andorgraph'$ is a finite sequence $v_0 = v, v_1$ \ldots $v_k$ of nodes of $\andorgraph'$, with 
$k \geq 0$, such that, for every $0 \leq i < k$, $v_i$ is an or-node and $(v_i,v_{i+1})$ is and edge of $\andorgraph'$, and $v_k$ is and-node.
\end{definition}

\noindent 
Since the end-node is an and-node,
all nodes  of $\andorgraph'$ have a saturation path.
\\\\

\noindent 
The domain $\Delta$ 
of the \graph{\setRoles}
obtained by applying the algorithm of Table \ref{table:algorithmbuildRstructured}, 
consists of elements  that also belong to 
 $\dom(\As) \cup \dom(\Mb)$, which we call {\em individuals}, and 
elements 
 that are created to meet the conditions of the knowledge base
 given by the existential restrictions,  
 which we call {\em variables}. 
\\

\noindent The function $f$  in the algorithm  of Table \ref{table:algorithmbuildRstructured}
maps elements in $\Delta$ to an and-node of the graph $\andorgraphP$.
\begin{enumerate}
\item If $a$ is an individual from the knowledge base then $f(a) = v_0$. 
The individual $a$  comes from the unique base node which is the initial
node $v_0$. In this case, the base node is also an and-node because
no base rule is applicable to the knowledge base except for $\transPR$.

\item If $y$ is a  variable (a newly created individual),
 then $f(y) = w_h$.
The individual $y$ comes from the last  variable node $w_h$ in the saturated path
which is an and-node.

\end{enumerate}

\begin{theorem}[Correctness of the Algorithm   of Table \ref{table:algorithmbuildRstructured}]
\label{theorem:constructionsaturatedstructure}
Let $\Kb = \KB{\Tb}{\Ab}{\Mb}$ be a knowledge base in $\ALCM$ in negation normal form such that 
\begin{enumerate}
\item $\Ab$ does not contain equalities,
\item  no base rule is applicable to $\Kb$  except for the $\transPR$-rule, 
\item 
 $C(d) \in \Ab$  for all $C \in \Tb$
 and  $d \in \dom(\Ab) \cup \dom(\Mb)$.
\end{enumerate}
Suppose  
$\andorgraph$ is an and-or graph for  
$\Kb = \KB{\Tb}{\Ab}{\Mb}$ with a consistent marking $\andorgraphP$.
Then, 
the structure $\T = (\setS, \lab, \Eps)$ 
obtained by applying the algorithm of Table~\ref{table:algorithmbuildRstructured} is a 
saturated \graph{\setRoles} for $\KB{\Tb}{\Ab}{\Mb}$. 
\end{theorem}

\begin{proof}
We first prove that $\T$ is a saturated \graph{\setRoles} for $\Tb$.
For this, we need to prove the six 
conditions of Definition
	\ref{definition:saturatedgraph}.
	Suppose
	 $x,y \in \setS$, $R \in \setRoles$ and  concepts  $C, C_1, C_2, D$.

	\begin{enumerate}
	 \item 
	 Suppose towards a contradiction that 
	 there is $x \in \Delta$ such that 
	 $\{ \neg C, C \} \subseteq \lab(x)$.
	 Since all concepts are in \NNF, $C$ is an atomic concept  $B$,
	  so $\{ \neg B, B  \} \subseteq \lab(x)$. We divide en two cases:
	 \begin{enumerate}
	  \item Suppose $x \in \dom(\Ab) \cup \dom(\Mb)$.\\
	  By construction, $\{ \neg B(x), B(x)  \} \subseteq \Ab$. 
	  But then we can apply the $\botOnePR$-rule which it is not possible 
	 by the second hypothesis for $\Kb$, namely no base rule is applicable to $\Kb$  except for the $\transPR$-rule.

	  \item Suppose $x \not \in \dom(\Ab) \cup \dom(\Mb)$. \\
	  By construction, $\lab(x) = \bigcup_{i=0}^{h} \Xs_{i}$ where $\derconcepts{\Ts}{\Xs_0}$, ..., $\derconcepts{\Ts}{\Xs_h}$ are the
	  labels of the saturation path $w_0$, \ldots, $w_h$, $f(x)=w_h$.
	  It is not possible that each concept is on a different node since $w_0$, \ldots, $w_h$ is a saturation path
	  and no rules eliminate an atomic concept of the label of a node.
	  Then,  $\{ \neg B, B \} \subseteq \Xs_i$ for some $i$ 
	  such that $1 \leq i \leq h$. 
	  By Definition \ref{definition:preferences}, 
	  the $\botR$-rule has higher priority than any of the other rules.
	  This means that $w_i$ should be the last node in the path, i.e.
	   $i=h$.
	  But then, $w_h$ would be an absurdity node which contradicts the
	  $\andorgraphP$ is a consistent marking.

	 \end{enumerate}

	 \item 
	 %
	  Suppose $C_1 \sqcap C_2 \in \lab(x)$. We split in two cases:
	   
	 \begin{enumerate}
	
	  \item Suppose  $x \in \dom(\Ab) \cup \dom(\Mb)$.\\
	  By construction, 
	  $C_1 \sqcap C_2 \in \lab(x)$ if $(C_1 \sqcap C_2)(x)  \in \Ab$.
	  Then we have that $\{C_1(x), C_2(x) \} \subseteq \Ab$ 
	  because the $\andPR$-rule is not applicable to $\Ab$.
	  Hence, $\{C_1, C_2 \} \subseteq \lab(x)$
	  
	  \item Suppose $x \not \in \dom(\Ab) \cup \dom(\Mb)$.\\
	  By construction, $\lab(x) = \bigcup_{i=0}^{h} \Xs_{i}$  
	 where 
	  $f(x) = w_h$ 
	 and  the labels of the saturation path 
	 $w_0$, \ldots, $w_h$ are 
	 $\derconcepts{\Ts}{\Xs_0}$, ..., $\derconcepts{\Ts}{\Xs_h}$.
	Hence, there exists $1 \leq i \leq h$ such that 
	$C_1 \sqcap C_2 \in \Xs_i$.
	We have that $C_1 \sqcap C_2 \not \in \Xs_h$ because 
	$w_h$ is an and-node and 	the $\andR$-rule is not applicable to
		 $\derconcepts{\Ts}{\Xs_h}$.   
		 Hence, there is an $i <h$ such that 
		 $C_1 \sqcap C_2  \in \Xs_i$, $C_1 \sqcap C_2 \not \in \Xs_{i+1}$
		 and $C_1,  C_2  \in \Xs_{i+1}$.
		 Hence, 
		   $C_1, C_2  \in \Xs_{i+1} \subseteq \lab(x)$.        
	 \end{enumerate}

	 \item Suppose $C_1 \sqcup C_2 \in \lab(x)$.
	 We divide in two cases:
	 
	 
	 \begin{enumerate}
	  \item Suppose $x \in \dom(\Ab) \cup \dom(\Mb)$.
	  \\
	  By construction, $C_1 \sqcup C_2 \in \lab(x)$ if $(C_1 \sqcup C_2)(x)  \in \Ab$.
	  \\
	  Then we have that $C_1(x) \in \Ab$ or $ C_2(x) \in \Ab$ because the $\orPR$-rule is not applicable to $\Ab$.
	  Hence, $C_1 \in \lab(x)$ or $C_2 \in \lab(x)$.

	  \item Suppose $x \not \in \dom(\Ab) \cup \dom(\Mb)$.
	  
	  By construction, $\lab(x) = \bigcup_{i=0}^{h} \Xs_{i}$  
	 where 
	  $f(x) = w_h$ 
	 and  the labels of the saturation path 
	 $w_0$, \ldots, $w_h$ are 
	 $\derconcepts{\Ts}{\Xs_0}$, ..., $\derconcepts{\Ts}{\Xs_h}$.
	
	$C_1 \sqcup C_2 \not\in \Xs_h$ because $w_h$ is and-node and the $\orR$-rule is not applicable to $\derconcepts{\Ts}{\Xs_n}$.
	Hence, there exists $1 \leq i < h$ such that $C_1 \sqcup C_2 \in \Xs_i$, $C_1 \sqcup C_2 \not\in \Xs_{i+1}$ and 
	$C_1 \in \Xs_{i+1}$ or $C_2 \in \Xs_{i+1}$. So $C_1 \in \lab(x)$ or $C_2 \in  \lab(x)$.

	 \end{enumerate}

	 \item 
			 Assume $x \in \Delta$, $\forall R.D \in \lab(x)$ 
			 and $(x,y) \in \Eps(R)$, we show that $D \in \lab(y)$.
			 
				\begin{enumerate}
				 \item Suppose $x \in \dom(\Ab) \cup \dom(\Mb)$.

					By construction, $\forall R.D(x) \in \Ab$.
					
					\begin{itemize}
				    \item Suppose $y \in \dom(\Ab) \cup \dom(\Mb)$.
						 
						  Then $R(x, y) \in \Ab$ since $(x,y) \in \Eps(R)$.
							Then we have that $D(y) \in \Ab$ because the
							 $\forallPR$-rule is not applicable to $\Ab$.
							Hence, $D \in \lab(y)$.

					\item Suppose $y \not \in \dom(\Ab) \cup \dom(\Mb)$	.
						
						By construction, 
						$\lab(y) = \bigcup_{i=0}^{h} \Xs_{i}$ 	 
						where the labels of the saturation path 
						$w_0$, \ldots, $w_h$ are 	 
						$\derconcepts{\Ts}{\Xs_0}$, ..., 
						$\derconcepts{\Ts}{\Xs_h}$.
	                    Since 
						$\forall R.D(x) \in \Ab$, 
						 it belongs to the label of $v_0$.
						  Since $R(x,y)$ cannot belong to $\Ab$
						  because $y \not \in \dom(\Ab)$, 
						 we have that 
						  $\exists R.C(x) \in \Ab$ for some concept $C$.
						  By 	application of the $\transPR$-rule we
						   know that $D$ belongs to the label of $w_0$ 
						   and hence also to $\lab(y)$.
							\end{itemize}
						 
					\item Suppose $x \not \in \dom(\Ab) \cup \dom(\Mb)$.		
					\\
					 Then, $y \not \in \dom(\Ab) \cup \dom(\Mb)$ 
					 because in $\Eps$ there are no pairs 
					 $(x,y)$ where $x \not \in \dom(\Ab) \cup \dom(\Mb)$ 
					 and $y \in \dom(\Ab) \cup \dom(\Mb)$. 
					 
						So by construction, we know that:
						\begin{itemize}
						 \item $\lab(x) = \bigcup_{i=0}^{h} \Xs_{i}$ 
						 	 where the labels of the saturation path 
						$w_0$, \ldots, $w_h$ are 	
						 $\derconcepts{\Ts}{\Xs_0}$, ..., 
						 $\derconcepts{\Ts}{\Xs_h}$.
						
						 \item $\lab(y) = \bigcup_{i=0}^{p} \Zs_{i}$ 
						 	 where the labels of the saturation path 
						$v_0$, \ldots, $v_p$ are 	 
						$\derconcepts{\Ts}{\Zs_0}$, ..., 
						$\derconcepts{\Ts}{\Zs_p}$.
						\end{itemize}
						
						Since $(x,y) \in \Eps(R)$ 
						 and $R(x,y) \not \in \Ab$, 
						we know that $\exists R.C_1 \in \lab(x)$						
						for some $C_1$.
						Since the rules applied to the nodes 
						$w_0$, \ldots, $w_{h-1}$ 
						do not  eliminate  
						for all and exists concepts, we know that
						$\{ \forall R.D, \exists R.C_1 \} \subseteq \Xs_h$.
						 The node $v_0$ is obtained from
						$w_h$ by application of the $\transR$-rule.
										 
						In $\andorgraphP$ there is an edge $(w_h,v_0)$
						 labelled $\exists R.C_1$.
						
						By the application of the $\transR$-rule 
						at the node $w_h$ we know that 
						$D \in \Zs_0$, so $D \in \lab(y)$.

				\end{enumerate}

	 \item 
	 Assume  $\exists R.C \in \lab(x)$ for $x \in \Delta$.
	
	At the step 2 of the algorithm we create (if it does not exist)
	a new variable $y$ and add $(x,y)$ to $\Eps(R)$.
	For this $y$, $\lab(y) = \bigcup_{i=0}^{h} \Xs_{i}$ 	 where the labels of the saturation path 
						$w_0$, \ldots, $w_h$ are 	 $\derconcepts{\Ts}{\Xs_0}$, ..., $\derconcepts{\Ts}{\Xs_h}$.
						
	Let $u = f(x)$ (the node associated with $x$ in $\andorgraphP$).
	 
	Then, $w_0$ is  obtained 
	from $u$  by applying the $\transR$ or $\transPR$-rule. 
	 So $C \in \Xs_0$ and also $C \in \lab(y)$.
	 
	 \item 
	%
	 We prove that $\Tb \subseteq \lab(x)$ for all $x \in \Delta$.
	We divide in cases:
	\begin{enumerate}
	  \item Suppose $x \in \dom(\Ab) \cup \dom(\Mb)$.
	  By initialization,
	  \[
	  \lab(x) = \{ C \mid C(a) \in \Ab \}
	  \]
	  Then, $\Tb \subseteq  \lab(x)$  by 
	  the third hypothesis on $\Kb$.
	   
	  \item Suppose $x \not \in \dom(\Ab) \cup \dom(\Mb)$.
	  By construction, $\lab(x) = \bigcup_{i=0}^{h} \Xs_{i}$  
	 where 
	  $f(x) = w_h$ 
	 and  the labels of the saturation path 
	 $w_0$, \ldots, $w_h$ are 
	 $\derconcepts{\Ts}{\Xs_0}$, ..., $\derconcepts{\Ts}{\Xs_h}$.
	 The node $w_0$ was obtained by application of 
	 either the $\transPR$-rule or $\transR$-rule.
	 Since both rules add the  Tbox to $\Xs_0$,
	  we have that
	 \[ 
	 \Tb \subseteq \Xs_0 \subseteq \lab(x)\]
	 
	  \end{enumerate}
	
	\end{enumerate}

	\noindent
We now prove 
that $\T$ is a saturated \graph{\setRoles} for $\KbALC{\Tb}{\Ab}$.
For this, it only remains to prove 
the four conditions of Definition \ref{definition:saturatedgraphALC}.

      \begin{enumerate}
      
	\item 
	By step 1 in the algorithm of Table~\ref{table:algorithmbuildRstructured}
	all the individuals of $\Ab$ belongs to $\Delta$.\\
      
	\item 
	 It holds by initialization of $\lab(a)$ (step 1 
	 in the algorithm) for each 
	$a \in \dom(\Ab) \cup \dom(\Mb)$ in the first step of algorithm, so if $C(a) \in \Ab$, then $C \in \lab(a)$. \\
	
	\item 
	It holds by initialization of $\Eps(R)$ for each role name $R$
	 (step 1 in the algorithm), so if
	$R(a, b) \in \Ab$, then  $(a, b) \in \Eps(R)$.\\
	
	\item 
	If $a \not = b  \in \Ab$ then $a$ and $b$ are syntactically different, otherwise we could apply 
	$\botTwoPR$-rule which it is not possible 
	 by the second hypothesis for $\Kb$.
	
 \end{enumerate}

We now prove that $\T$ is a 
\graph{\setRoles} 
for $\ALCM$. For this, we only need to prove 
the four conditions of  
Definition \ref{definition:saturatedgraphALCM}.

\begin{enumerate}
 \item 
 By step 1 of the algorithm all the individuals of $\Mb$ belongs to $\Delta$.\\
  
 \item 
 Suppose towards a contradicion that
  $\T$ has circularities w.r.t. a $\Mb$. 
  Then  there is   
  a sequence of meta-modelling axioms 
  $a_1 \eqm A1, a_2 \eqm A2 \ldots a_n \eqm A_n$ all in $\Mb$
  such that $A_1 \in \lab(a_2), A_2 \in \lab(a_3) \ldots A_n \in \lab(a_1)$.
 That is, $A_1(a_2)$, $A_2(a_3)$, \ldots, $A_n (a_1)$ are in $\Ab$, so from Definition \ref{definition:circular} we have $\circular(\Ab, \Mb)$.
 But then we can apply the $\botThreePR$-rule which  
 contradicts the second hypothesis for $\Kb$.   \\
 
 \item 
 Suppose $a \eqm A\in \Mb$ and $a \eqm B\in \Mb$.
 Then,   $A  = B$. Otherwise we could apply 
 the $\equal$-rule which it is not possible 
 by  the second hypothesis for $\Kb$. \\
 
 \item 
 Suppose $a$ and $b$ are syntactically different, 
 $a \eqm A\in \mathcal{M}$ and $b \eqm B\in \mathcal{M}$.
 Then there is some $t \in \setS$ such that $A \sqcap \neg B \sqcup B \sqcap \neg A \in \lab(t)$.
 Otherwise we could  apply the $\difference$-rule which it is not possible    by  the second hypothesis for $\Kb$. 
 \end{enumerate}

\end{proof}

\begin{theorem}[Completeness of the Tableau Calculus of $\ALCM$]
\label{theorem:completeness}
Let $\Kb = \KB{\Tb}{\Ab}{\Mb}$
be a knowledge base of $\ALCM$ in negation normal form.
If the and-or graph for  $\Kb$
 has a consistent marking
 then
 $\Kb$ is consistent.
\end{theorem}

\begin{proof} Suppose $\andorgraphP$ is a consistent marking of the and-or graph $\andorgraph$ of $\Kb$. 
Let $v_0, v_1, \ldots, v_n$ be 
the saturation path from the root of $\andorgraphP$.
  It is clear that 
every $v_i$ is a base node for all all $1 \leq i \leq n-1$.
The last node cannot be an absurdity node because
$\andorgraphP$ a consistent marking. Hence, $v_n$ should
be a base node as well. 
This means that the node $v_i$ has label 
\[
\derboxes{\Tb_i}{\Ab_i}{\Mb_i}
\]
and in particular the label of  $v_0$ is
\[
 \derboxes{\Tb_0}{\Ab_0}{\Mb_0}
\]
where 
\[
\begin{array}{ll}
\Tb_0 & := \Tb \\
\Ab_0 & := \Ab^{*} \cup 
\{C(a) \ \suchthat C \in \Tb, \  a  \in \dom(\Ab^{*}) \cup \dom(\Mb^{*}) \}  \\
\Mb_0 & := \Mb^{*} \\
\end{array}
\]

where $\Ab^{*}$ and $\Mb^{*}$ are obtained from $\Ab$ and $\Mb$
 by choosing a canonical representative $a$ for each assertion $a = b$ and replacing $b$ by $a$. 
\\
 We will 
apply the algorithm of Table \ref{table:algorithmbuildRstructured} 
to $\derboxes{\Tb_n}{\Ab_n}{\Mb_n}$. For this we need to
prove that the hypotheses of  Theorem \ref{theorem:constructionsaturatedstructure} hold:
\begin{enumerate}

\item No base rule is applicable to
$\derboxes{\Tb_n}{\Ab_n}{\Mb_n}$ except for
$\transPR$.  To prove this we divide in cases:

\begin{itemize}

\item Suppose $v_n$ is an end-node. Then, 
it follows trivially that  
no rule is applicable to the node.

\item Suppose that the rule $\transPR$ was applied to
 $v_n$ to obtain its successors in the and-or graph $\andorgraph$.
It follows from 
 Definition \ref{definition:preferences}
 that
the $\transPR$-rule could be applied to $v_n$  
only if the other rules are not applicable.
\end{itemize}
\item 
 It is easy to see that $\Ab_n$ does not contain
equalities  since
$\Ab_0$  does not contain equalities.

\item We also have that 
 $C(d) \in \Ab$ for all $C \in \Tb$ and
  $d \in \dom(\Ab_n) \cup \dom(\Mb_n)$.
  This  property holds in the initialization
  and it is also preserved after applying any of the base rules.
  Note that the only rule we could apply in the saturation path
  $v_0, \ldots, v_n$ that adds new 
  individuals is 
  the $\difference$-rule.   
\end{enumerate}
\noindent
It follows from Theorem \ref{theorem:constructionsaturatedstructure}
that there exists a a saturated \graph{\setRoles}
for $\derboxes{\Tb_n}{\Ab_n}{\Mb_n}$.
By Theorem \ref{theorem:saturatedgraphALCM},
$\derboxes{\Tb_n}{\Ab_n}{\Mb_n}$ is satisfiable.
Applying the Lemma~\ref{lemma:conversepreservation} we know that
$\derboxes{\Tb_0}{\Ab_0}{\Mb_0}$ is satisfiable and
from  Lemma~\ref{lemma:satisfiableSust}  that
$\derboxes{\Tb}{\Ab}{\Mb}$ is satisfiable too.

\end{proof}

\section{An ExpTime Decision Procedure for $\ALCM$}
\label{Section:ExpTimeAlgorithm}

In this section, we prove that the complexity does not increase when we move
from $\ALC$ to $\ALCM$.
  In order to prove this, it is enough to 
   give an algorithm for checking consistency of a knowledge base
 in $\ALCM$ that is ExpTime. 
\\\\
\noindent 
Let $\Kb = \KB{\Tb}{\Ab}{\Mb}$ a knowledge base in $\ALCM$ with $\Tb$ and $\Ab$ in negation normal form (NNF). We claim that the \emph{Algorithm for Checking Consistency in $\ALCM$} given 
in Table \ref{table:algorithmALCM}
 is an ExpTime (complexity-optimal) algorithm for checking consistency of $\Kb$. 
 In the algorithm, a node $u$ is a parent of a node $v$ and $v$ is a child of $u$ iff 
 the edge $(u, v)$ is in  $\andorgraph$. 
 Note also that there is a unique node with label
$\derbot$.

\begin{table}

\setlength{\parindent}{0pt}\line(1,0){345}
\newline
{\bfseries{Algorithm for Checking Consistency in $\ALCM$}} 
\newline\newline
{\bfseries{Input}}:  $\Kb = \KB{\Tb}{\Ab}{\Mb}$  in negation normal form.
\newline
{\bfseries{Output}}: $true$ if $\Kb$ is consistent, and $false$ otherwise.
\newline
\setlength{\parindent}{0pt}\line(1,0){345}

\begin{compactenum}
\item  Construct and ``and-or'' graph  $\andorgraph$
  with root $v_0$ for $\KB{\Tb}{\Ab}{\Mb}$;  
\item UnsatNodes := $\emptyset$, $U := \emptyset$;
\item If   $\andorgraph$ 
 $contains$ $a$ $node$ $v_{\bot}$ $with$ $label$ 
 $\derbot$ 
then
\begin{compactdesc}
\item  $U$ := $\{v_{\bot}\}$,  UnsatNodes := $\{v_{\bot}\}$;

\item  while $U$ $is$ $not$ $empty$ do

\begin{compactdesc}

\item  remove a node $v$ from $U$;

\item for every parent $u$ of $v$ do

\begin{compactdesc}
\item  if $u \notin$ UnsatNodes $and$ ($u$ $is$ $an$ \emph{and-node} $or$
$u$ $is$ $an$ \emph{or-node} $and$ every child of $u$ $is$ $in$  
 $UnsatNodes$) then
add $u$ to both UnsatNodes and $U$
\end{compactdesc}

\end{compactdesc}

\end{compactdesc}

\item  return $false$ if $v_0 \in$ UnsatNodes, and $true$ otherwise

\end{compactenum}

\setlength{\parindent}{0pt}\line(1,0){345}

\caption{Algorithm for checking consistency in $\ALCM$}

\label{table:algorithmALCM}

\end{table}

\noindent
Recall that a $formula$ is 
either a concept, or an Abox-statement or an Mbox-statement. 
We define the \emph{length} of a formula to be the number of its symbols, and the $size$ of a finite set of formulas to be the sum of the lenghts of its elements.

\begin{lemma} \label{lemma:cantNodesGraph}
Let $\Kb = \KB{\Tb}{\Ab}{\Mb}$ be an $\ALCM$-knowledge base in negation normal form, 
$n$ be the size of $\Tb \cup \Ab \cup \Mb$, and $\andorgraph$ be an and-or graph for $\Kb$. 
Then $\andorgraph$ has $O(2^{n^4})$ nodes.
\end{lemma}

\begin{proof}
For each pair $a, b$ of individuals with meta-modelling in the Mbox, the algorithm either adds a new TBox axiom, using  the $\equal$-rule, or adds an individual that is denoted by $d_{a, b}$, using the $\difference$-rule.
We define
\[
\begin{array}{ll}
\Tb^{+}   = 
 \Tb \cup 
  \{ (A \sqcup \neg B) \sqcap ( B \sqcup \neg A)
 \mid 
 a \eqm A, b \eqm B \in \Mb \}\\
 \Ab^{+}  =  \Ab \cup 
 \{  (A \sqcap \neg B \sqcup  B \sqcap \neg A)(d_{a,b})  \mid 
 a \eqm A, b \eqm B \in \mathcal{M}
 \} 
 \end{array}
\]
%

The sets  $\Tb^{+}$ and  $ \Ab^{+}$ have  cardinality
 $O(n^2)$ since 
\begin{equation}
\label{eq:cardinality}
\#\{(a,b) \mid a, b \in \dom(\mathcal{M})\} \leq n^2
\end{equation}
The label of each base  node  $v$
of $\andorgraph$
is $\derboxes{\Tb_v}{\Ab_v}{\Mb_v}$.
The sets $\Tb_v$, $\Ab_v$ and $\Mb_v$ 
have the following upper bounds:
\[
\begin{array}{lll}
\Tb_v & \subseteq & \Tb^{+} \\ \\
\Ab_v &  \subseteq &  \{ D(b) \mid C(a) \in \Ab^{+},  
                          D \in \SC(C) \mbox{ and }  b \in \dom(\Ab^{+}) \cup \dom(\Mb) \} \cup \\
         & &   \{ C(a) \mid C \in \SC(\Tb^{+}) \mbox{ and }
                         a \in \dom(\Ab^{+}) \cup \dom(\Mb) \} \cup \\
       & &  
             \{ a \not = b \mid a \not =b \in \Ab 
              \mbox{ or } a, b \in \dom(\Mb)  \}
\\\\
\Mb_v & \subseteq & \{ a \eqm A \mid a \in \dom(\Mb), A \in {\range}(\Mb) \} 
\end{array}
\]

where $\SC(\Tb^{+})$ denotes the image of $\Tb^{+}$ 
under $\SC$, i.e.
$\SC (\Tb^{+}) = \{ \SC(C) \mid C \in \Tb^{+} \}$.
\\
Recall $\SC(C)$ gives the set of subconcepts of $C$. 
\\
The upper bound for $\Ab_v$
is the union of three sets: 
 \[
 \begin{array}{ll}
 \{ D(b) \mid C(a) \in \Ab^{+},  
                          D \in \SC(C) \mbox{ and }  b \in \dom(\Ab^{+}) \cup \dom(\Mb) \} \cup \\
         \{ C(a) \mid C \in \SC(\Tb^{+}) \mbox{ and }
                         a \in \dom(\Ab^{+}) \cup \dom(\Mb) \} \cup \\
       
             \{ a \not = b \mid a \not =b \in \Ab
             \mbox{ or } a, b \in \dom(\Mb) \}
             
             \end{array}
 \]
The first set $\{ D(b) \mid C(a) \in \Ab^{+},  
                          D \in \SC(C) \mbox{ and }  b \in \dom(\Ab^{+}) \cup \dom(\Mb) \}$
                          includes the axioms $C(b)$ that are obtained from replacing $a$ by $b$
                          in $C(a) \in \Ab$.
\\
The cardinality of the first set has
$O(n^4) = O(n^2 ) \times O(n^2)$.
This is because we are combining $D$'s with $b$'s.
The number of $D$'s as well as the number of $b$'s
 have $O(n^2)$
 since in both cases the cardinality of $\Ab^{+}$
is the predominant one.
  \\
Similarly, the cardinality of the second set
 has
$O(n^4) = O(n^2 ) \times O(n^2)$.
This is because we are combining $C$'s with $a$'s.
There are as many $C$'s as elements in
$\SC(\Tb^{+})$ and the latter  has $O(n^2)$.
The number of $a$'s has the same order as the cardinality
of $\Ab^{+}$ which is $O(n^2)$.
\\
It follows from \eqref{eq:cardinality} that
the third set has cardinality  $O(n^2)$. 
\\
Hence, the cardinality of the upper bound for $\Ab_v$
which is the union of these three sets
has $O(n^4)$.
\\
\\
The number of base nodes  has the following order:
\[
O(2^{n^2}) \times  O(2^{n^4}) \times  O(2^{n^2}) = O(2^{n^4})
\]
which is the multiplication of the
number of subsets of the upper bounds for
$\Tb_v$, $\Ab_v$ and $\Mb_v$. 
\\
\\
The label of each variable  node  $w$
of $\andorgraph$
is $\derconcepts{\Tb_w}{\Xs_w}$
where
\[
\begin{array}{lll}
\Tb_w & \subseteq & \Tb^{+} \\ \\
\Xs_w &  \subseteq &  \{ D \mid C(a) \in \Ab^{+} \mbox{ and } 
                          D \in \SC(C) \} \cup \\
         & &   \{ C \mid C \in \SC(\Tb^{+})  \} 
\end{array}
\]

The cardinality of the bound for $\Xs_w$ 
 has $O(n^2)$.
Hence, the number of variables nodes are
\[
O(2^{n^2}) \times  O(2^{n^2}) = O(2^{n^2})
\]
which is the multiplication of the
number of subsets of the upper bounds for
$\Tb_w$ and $\Xs_w$. 
\end{proof}

\begin{table}

\setlength{\parindent}{0pt}\line(1,0){344}
\newline
{\bfseries{Algorithm for Checking Circularities}} 
\newline
\setlength{\parindent}{0pt}\line(1,0){344}
\newline
{\bfseries{Input}}: an Abox $\Ab$ and an Mbox $\Mb$.
\newline
{\bfseries{Output}}: $true$ if there exists a circularity in $\Ab$ w.r.t. $\Mb$ and $false$ otherwise.

\begin{enumerate}
\item  Given the Abox $\Ab$ and Mbox $\Mb$, construct a directed graph as follows.
\begin{enumerate}
\item The nodes are the 
elements in $\dom(\Mb)$.

\item
There is an edge from $a$ to $b$ if $B(a) \in \Ab$ and $b \eqm B$. 
\end{enumerate}

\item 
 Check if there is a cycle in the graph constructed in the previous part using some well-known  algorithm 
for cycle detection, e.g. Chapter 4 of \cite{SedgewickWayne-Algorithms}.

\end{enumerate}

\caption{Algorithm for Checking Circularities}

\label{table:algorithmcircularities}

\end{table}

The construction of the and-or graph needs to 
check if each node has circularities or not.
We show an algorithm for checking circularities 
in Table \ref{table:algorithmcircularities}.
For example, consider the Mbox 
\[
  \begin{array}{llll}    
  a_0  \eqm  A_0 & 
      a_1  \eqm  A_1 &
       a_2  \eqm  A_2 &
     a_3  \eqm  A_3
     \end{array}
\]
and the Abox
\[
\begin{array}{lllll}
       A_1(a_0) &  A_0(a_2) & A_3(a_2) & A_2(a_1)
       \end{array}
\]
We construct the graph illustrated in Figure \ref{fig:cycles}, whose nodes are $a_0, a_1, a_2$ and $a_3$.
In this graph, there is an edge from the node 
$a_i$ to $a_j$
if  $(a_i)^{\interp} \in (a_j)^{\interp}$ for a model $\interp$. 
In other words, the edges represent the membership relation
$\in$.

\begin{figure}
\centering
\includegraphics[width=0.5\linewidth]{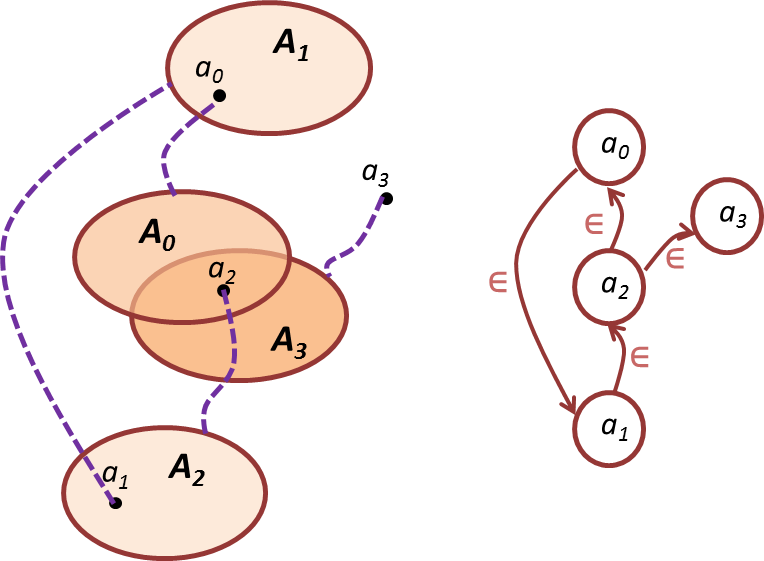}
\caption{Knowledge base with cycles and associated directed graph  }
\label{fig:cycles}
\end{figure}

\begin{lemma} \label{lemma:terminationGraph}
The \emph{Algorithm for Checking Consistency in $\ALCM$} of Table \ref{table:algorithmALCM} terminates and computes the set UnsatNodes in  $O(2^{n^4})$ steps  
 where $n$ is the size of $\Tb \cup \Ab \cup \Mb$. 
\end{lemma}

\begin{proof}

Each node is processed in $O(n)$ since
checking for clashing and circularities take $O(n)$.
For checking circularities, we need to detect cycles
in a graph (the edges represent the membership relation 
$\in$) which takes $O(n)$ \cite{SedgewickWayne-Algorithms}
(see Table  \ref{table:algorithmcircularities}). 
 Lemma  \ref{lemma:cantNodesGraph} 
  guarantees that the and-or graph $\andorgraph$ can be built in $O(2^{n^4})$.
Every node put into $U$ is also put into UnsatNodes, but once a node is in UnsatNodes, it never leaves UnsatNodes and cannot be put back into $U$. 
Each iteration of the ``while'' removes one member of $U$. Since the number of nodes in $\andorgraph$ is $O(2^{n^4})$, this means that after 
at most $O(2^{n^4})$ 
iterations, $U$ become empty. Each iteration is done in $O(2^{n^4})$ steps.
 Hence the algorithm terminates after $O(2^{n^4})$ steps.
 \end{proof}

\begin{theorem}
\label{theorem:correctnessandcomplexityofalgorithm}
The algorithm of Table \ref{table:algorithmALCM} is an ExpTime
decision procedure for checking consistency of a knowledge base
in $\ALCM$.
\end{theorem}

\noindent
\begin{proof}
The proof of correctness is the same as of \cite[Theorem 5.3]{DBLP:journals/jar/GoreN13}. 
We have to use
Theorem \ref{theorem:correctnesstableau}. 
Complexity follows from Lemma \ref{lemma:terminationGraph}.
\end{proof}
\noindent
We can now show the main new result of this paper:

\begin{corollary}[Complexity of $\ALCM$]
Consistency of a (general) knowledge base in  $\ALCM$ is ExpTime-complete.
\end{corollary}

Hardness follows from the corresponding result  for $\ALC$ (see Theorem \ref{theorem:complexityALC}).
A matching upper bound for $\ALCM$ is given by 
the algorithm of Table \ref{table:algorithmALCM} which by 
Theorem \ref{theorem:correctnessandcomplexityofalgorithm}
is ExpTime.

\section{Related Work}

We made several changes to the
ExpTime tableau algorithm 
for $\ALC$ by Nguyen and Szalas 
 to accommodate
meta-modelling~\cite{Nguyen2009}.
First of all, we added 
some rules for dealing
  with the equalities and inequalities that need to be transferred
from the Tbox to the Abox and vice versa.
There is also a new rule that returns inconsistency 
in case a  circularity is found. This new rule is key 
for our approach to meta-modelling
and  ensures that the domain of the interpretation is  well-founded.
Since our tableau algorithm  has the peculiarity of changing the TBox, 
the Tbox (and also the Mbox) have to be stored in the labels of the and-or graph,
 a fact that was not necessary in the simpler Tableau Calculus for $\ALC$ of Nguyen and Szalas 
 \cite{Nguyen2009}. 
 \\
In the literature of Description Logic,  
there are other approaches to  meta-modelling  \cite{Motik07,DBLP:conf/owled/PanHS05,DBLP:conf/semweb/GlimmRV10,DBLP:journals/ijsi/JekjantukGP10,DBLP:conf/aaai/GiacomoLR11,DBLP:conf/dlog/HomolaKSV13,DBLP:conf/dlog/HomolaKSV14,Lenzerinietal2014}.
The approaches which define fixed layers or levels of meta-modelling \cite{DBLP:conf/owled/PanHS05,DBLP:journals/ijsi/JekjantukGP10,DBLP:conf/dlog/HomolaKSV13,DBLP:conf/dlog/HomolaKSV14}  
impose a very strong limitation to the ontology engineer.
Our approach allows the user to have any number of levels or layers 
(meta-concepts, meta meta-concepts and so on). 
Besides the benefits of not having to know the layer of each concept and having the flexibility of mixing different layers, there is a more pragmatic advantage which arises from ontology engineering. 
In a real scenario of evolving ontologies, that need to be
integrated,  
 not all individuals of a given concept need to have
 meta-modelling and hence, they do not have to belong to the
 same level in the hierarchy. \\ 
The key feature in our semantics is to interpret $a$ and $A$
as the same object when $a$ and $A$ are connected through meta-modelling, i.e.,
if  $a \eqm A$ then $a^\interp = A^\interp$. 
 This allows us to detect  inconsistencies in  the ontologies 
  which is not possible under the 
  Hilog 
  semantics~\cite{Motik07,DBLP:conf/aaai/GiacomoLR11,DBLP:conf/dlog/HomolaKSV13,DBLP:conf/dlog/HomolaKSV14,Lenzerinietal2014,DBLP:conf/dlog/KubincovaKH15}. 
Our semantics also requires that
the domain of the interpretation be a  well-founded set. 
A domain such as $\Delta^{\interp} = \{X\}$ where $X = \{X\}$ is a set
that belongs to itself,  it  cannot represent any
 real object from  our  usual applications in Semantic Web.
 

\section{Conclusions and Future Work}

The ExpTime algorithm for $\ALCM$ presented in this paper  
can be optimized in several ways.
 Instead of constructing  first the and-or graph and then
checks whether the graph contains a consistent marking, 
we can do these two tasks simultaneously \cite{DBLP:journals/jar/GoreN13}.
Adding $(A \sqcup \neg B) \sqcap (B \sqcup \neg A)$ 
to the Tbox is not efficient since it generates too many expansions with
or-branching. We can instead  add $A \equiv B$ and apply  optimizing 
techniques of  lazy unfolding \cite{Horrocks03}.
\\
We plan to extend this algorithm to include 
other logical constructors such as 
cardinality restrictions, role hierarchies
and nominals  
\cite{DBLP:conf/dlog/DingH07,DBLP:journals/fuin/NguyenG14}.
\\
It is also possible to show Pspace-completeness for $\ALCM$
under certain conditions of unfoldable Tboxes. The details 
of this proof will appear in  a separate report.

\paragraph{Acknowledgements.}
The third author would like to acknowledge a Daphne Jackson fellowship
 sponsored by EPSRC and the University of Leicester.
 We would also like to thank Alfredo Viola for 
 some excellent suggestions.


\begin{thebibliography}{10}

\bibitem{Motz2015}
R. Motz, E. Rohrer, and P.Severi.
\newblock The description logic \emph{SHIQ} with a flexible meta-modelling
  hierarchy.
\newblock {\em Journal of Web Semantics: Science, Services and Agents on the
  World Wide Web}, 2015.

\bibitem{DBLP:conf/jist/MotzRS14}
R.~Motz, E.~Rohrer, and P.~Severi.
\newblock Reasoning for \emph{ALCQ} extended with a flexible meta-modelling
  hierarchy.
\newblock In {\em 4th Joint International Semantic Technology Conference,
  {JIST} 2014}, volume 8943 of {\em Lecture Notes in Computer Science}, pages
  47--62, 2014.

\bibitem{DBLP:conf/ijcai/Schild91}
K.~Schild.
\newblock A correspondence theory for terminological logics: Preliminary
  report.
\newblock In {\em Proceedings of the 12th International Joint Conference on
  Artificial Intelligence}, pages 466--471, 1991.

\bibitem{DBLP:conf/dlog/GiacomoL96}
G.~{De Giacomo} and M.~Lenzerini.
\newblock {TBox} and {ABox} {R}easoning in {E}xpressive {D}escription {L}ogics.
\newblock In {\em Proceedings of Description Logic Workshop}, pages 37--48,
  1996.

\bibitem{DBLP:conf/dlog/2003handbook}
F.~Baader, D.~Calvanese, D.~L. McGuinness, D.~Nardi, and P.~F. Patel-Schneider,
  editors.
\newblock {\em The Description Logic Handbook: Theory, Implementation, and
  Applications}. Cambridge University Press, 2003.

\bibitem{DBLP:journals/ws/SirinPGKK07}
Evren Sirin, Bijan Parsia, Bernardo~Cuenca Grau, Aditya Kalyanpur, and Yarden
  Katz.
\newblock Pellet: {A} practical {OWL-DL} reasoner.
\newblock {\em Journal of Web Semantics: Science, Services and Agents on the
  World Wide Web}, 5(2):51--53, 2007.

\bibitem{Motik07}
B.~Motik.
\newblock On the properties of metamodeling in {OWL}.
\newblock {\em Journal of Logic and Computation}, 17(4):617--637, 2007.

\bibitem{DBLP:conf/owled/PanHS05}
J.~Z. Pan, I.~Horrocks, and G.~Schreiber.
\newblock {OWL FA}: A metamodeling extension of {OWL DL}.
\newblock In {\em {OWLED}}, volume 188 of {\em {CEUR} Workshop Proceedings},
  2005.

\bibitem{DBLP:conf/semweb/GlimmRV10}
B.~Glimm, S.~Rudolph, and J.~V{\"o}lker.
\newblock Integrated metamodeling and diagnosis in {OWL} 2.
\newblock In {\em International Semantic Web Conference (1)}, pages 257--272,
  2010.

\bibitem{DBLP:journals/ijsi/JekjantukGP10}
N.~Jekjantuk, G.~Gr{\"o}ner, and J.~Z. Pan.
\newblock Modelling and reasoning in metamodelling enabled ontologies.
\newblock {\em International Journal Software and Informatics}, 4(3):277--290,
  2010.

\bibitem{DBLP:conf/aaai/GiacomoLR11}
G.~De Giacomo, M.~Lenzerini, and R.~Rosati.
\newblock Higher-order description logics for domain metamodeling.
\newblock In {\em Proceedings of the Twenty-Fifth {AAAI} Conference on
  Artificial Intelligence, {AAAI} 2011}. {AAAI} Press, 2011.

\bibitem{DBLP:conf/dlog/HomolaKSV13}
M.~Homola, J.~Kluka, V.~Sv{\'{a}}tek, and M.~Vacura.
\newblock Towards typed higher-order description logics.
\newblock In {\em Proceedings of the 26th International Workshop on Description
  Logic}, pages 221--233, 2013.

\bibitem{DBLP:conf/dlog/HomolaKSV14}
M.~Homola, J.~Kluka, V.~Sv{\'{a}}tek, and M.~Vacura.
\newblock Typed higher-order variant of {SROIQ} - why not?
\newblock In {\em Informal Proceedings of the 27th International Workshop on
  Description Logics}, pages 567--578, 2014.

\bibitem{Lenzerinietal2014}
M.~Lenzerini, L.~Lepore, and A.~Poggi.
\newblock Making metaquerying practical for {Hi}({DL}−{LiteR}) knowledge
  bases.
\newblock In {\em On the Move to Meaningful Internet Systems: {OTM} 2014
  Conferences-Confederated International Conferences: CoopIS, and {ODBASE}
  2014}, volume 8841 of {\em Lecture Notes in Computer Science}, pages
  580--596. Springer, 2014.

\bibitem{Nguyen2009}
L.~A. Nguyen and A.~Szalas.
\newblock {ExpTime} tableaux for checking satisfiability of a knowledge base in
  the description logic {ALC}.
\newblock In {\em First International Conference, {ICCCI} 2009}, volume 5796 of
  {\em Lecture Notes in Computer Science}, pages 437--448, 2009.

\bibitem{FOST}
P.~Hitzler, M.~Kr{\"o}tzsch, and S.~Rudolph.
\newblock {\em Foundations of Semantic Web Technologies}.
\newblock Chapman \&{} Hall/CRC, 2009.

\bibitem{DBLP:journals/ai/Schmidt-SchaussS91}
M.~Schmidt{-}Schau{\ss} and G.~Smolka.
\newblock Attributive concept descriptions with complements.
\newblock {\em Artificial Intelligence}, 48(1):1--26, 1991.

\bibitem{DBLP:conf/dlog/GiacomoDM96}
G.~{De Giacomo}, F.~M. Donini, and F.~Massacci.
\newblock Exptime tableaux for {ALC}.
\newblock In {\em Proceedings of Description Logic Workshop}, pages 107--110,
  1996.

\bibitem{DBLP:journals/ai/DoniniM00}
F.~M. Donini and F.~Massacci.
\newblock {EXPTIME} tableaux for {ALC}.
\newblock {\em Artificial Intelligence}, 124(1):87--138, 2000.

\bibitem{DBLP:journals/jar/GoreN13}
R.~Gor{\'{e}} and L.~A. Nguyen.
\newblock Exptime tableaux for {ALC} using sound global caching.
\newblock {\em Journal of Automated Reasoning}, 50(4):355--381, 2013.

\bibitem{Winskel2010}
G.~Winskel.
\newblock {Notes and Exercises for Set Theory for Computer Science (pdf ) an
  MPhil course in Advanced Computer Science at the University of Cambridge.
  (Last data accessed January 2015)}.

\bibitem{DBLP:journals/jolli/Akman97}
J.~{Barwise} V.~Akman and L.~{Moss}.
\newblock {V}icious {C}ircles: {O}n the mathematics of {N}on-{W}ellfounded
  {P}henomena.
\newblock {\em Journal of Logic, Language and Information}, 6(4):460--464,
  1997.

\bibitem{Horrocks03}
I.~Horrocks.
\newblock Implementation and optimisation techniques.
\newblock In F.~Baader, D.~Calvanese, D.~L. McGuinness, D.~Nardi, and P.~F.
  Patel-Schneider, editors, {\em The Description Logic Handbook: Theory,
  Implementation, and Applications}, pages 306--346. Cambridge University
  Press, 2003.

\bibitem{Horrocks1999}
I.~Horrocks and U.~Sattler.
\newblock A description logic with transitive and inverse roles and role
  hierarchies.
\newblock {\em J Logic Computation}, 9(3):385--410, 1999.

\bibitem{SedgewickWayne-Algorithms}
R.~Sedgewick and K.~Wayne.
\newblock {\em Algorithms, 4th Edition}.
\newblock Addison-Wesley, 2011.

\bibitem{DBLP:conf/dlog/KubincovaKH15}
P.~Kubincov{\'{a}}, J.~Kluka, and M.~Homola.
\newblock Towards expressive metamodelling with instantiation.
\newblock In {\em Proceedings of the 28th International Workshop on Description
  Logics}, 2015.

\bibitem{DBLP:conf/dlog/DingH07}
Y.~Ding and V.~Haarslev.
\newblock An {ExpTime} tableau decision procedure for {ALCQI}.
\newblock In {\em Proceedings of Description Logic Workshop}, 2007.

\bibitem{DBLP:journals/fuin/NguyenG14}
L.~A. Nguyen and J.~Golinska{-}Pilarek.
\newblock An {ExpTime} tableau method for dealing with nominals and qualified
  number restrictions in deciding the description logic {SHOQ}.
\newblock {\em Fundamenta Informaticae - Concurrency Specification and
  Programming 2013}, 135(4):433--449, 2014.

\end{thebibliography}

\end{document}